\documentclass[10pt]{article}
\usepackage[a4paper, margin=1.25in]{geometry}
\usepackage{amsmath, amssymb, amsthm}
\usepackage{rotating}
\usepackage{graphicx,comment}
\usepackage{color}
\usepackage{float}
\usepackage{caption}
\usepackage{subcaption}
\usepackage[round, authoryear]{natbib}
\usepackage{authblk}
\pdfoutput=1
\makeatletter
\let\ps@plain\ps@empty
\makeatother

\newtheorem{theorem}{Theorem}[section]

\def\grad{{\bf \nabla}}
\def\div{{\bf \nabla\cdot}}

\begin{document}

\title{Stability Results on Radial Porous Media and Hele-Shaw Flows with Variable Viscosity Between Two Moving Interfaces}

\author[1]{Craig Gin}
\author[2]{Prabir Daripa}
\affil[1]{University of Washington, Department of Applied Mathematics}
\affil[2]{Texas A\&M University, Department of Mathematics}

\date{}

\renewcommand\Authands{ and }

\maketitle

\begin{abstract}
We perform a linear stability analysis of three-layer radial porous media and Hele-Shaw flows with variable viscosity in the middle layer. A nonlinear change of variables results in an eigenvalue problem that has time-dependent coefficients and eigenvalue-dependent boundary conditions. We study this eigenvalue problem and find upper bounds on the spectrum. We also give a characterization of the eigenvalues and prescribe a measure for which the eigenfunctions are complete in the corresponding $L^2$ space. The limit as the viscous gradient goes to zero is compared with previous results on multi-layer radial flows. We then numerically compute the eigenvalues and obtain, among other results, optimal profiles within certain classes of functions.
\end{abstract}

\section{Introduction}
Saffman-Taylor instability is an instability that occurs when a less viscous fluid drives a more viscous fluid in a porous medium. This phenomenon occurs in many applications including carbon sequestration, filtration, hydrology, and petrology. One important application is oil recovery in which it is a common practice to use water to displace oil. Oil recovery was the application driving the seminal work on this type of instability by \cite{Saffman/Taylor:1958}. In order to simplify their experiments, Saffman and Taylor studied the instability in the context of Hele-Shaw flows - flows between two parallel plates with a small gap between them. Hele-Shaw flows are a good model of porous media flows because they are also governed by Darcy's Law. It is now well known that a positive interfacial viscous jump in the direction of rectilinear flow produces an unstable flow with interfacial tension stabilizing short waves.

In applications in which it is advantageous to suppress the instability, one strategy is to use intermediate fluids which are more viscous than the displacing fluid but less viscous than the displaced fluid in order to limit the size of the viscous jumps. In the case that the interfacial tension at the interfaces is comparable to that of two-layer flow, this strategy is effective. This has been demonstrated by \cite{daripa08:multi-layer} in the case of rectilinear flow using upper bounds on the growth rates which were also used to give stabilization criteria.

Another approach to limiting the instability is using a variable viscosity fluid as the intermediate fluid. 
This can be achieved, for example, in chemical enhanced oil recovery in which polymer is used to increase the viscosity of the displacing fluid and variable polymer concentration leads to variable viscosity. A viscous gradient allows for even smaller viscous jumps at the interfaces than a constant viscosity fluid but the layer itself becomes unstable. There are several studies on the stability of multi-layer variable viscosity porous media and Hele-Shaw flows in a rectilinear geometry. \cite{GH83:optimal} first theoretically studied the stability of three-layer flows in such a geometry in which the middle fluid has variable viscosity. However, they studied the restricted case in which the trailing interface is a miscible interface with no interfacial tension. \cite{dp06:gradient} dropped this restriction and studied the case of three-layer variable viscosity flow with both interfaces immiscible. 

Saffman-Taylor instability is also studied in a radial flow geometry in which there is a point source in the center of the flow and the fluid moves outward radially with circular interfaces between fluids. 
Radial flow is one of several cases investigated by \cite{Muskat:1934, Muskat46:homogeneous} for the two-layer problem but in the case of zero interfacial tension. \cite{Paterson:1981} later performed a linear stability analysis for two-layer radial flow with interfacial tension. There are relatively few works on the stability of {\bf multi-layer} radial flows. \cite{Cardoso/Woods:1995} studied the stability of three-layer radial flows in the limiting case in which the inner interface is completely stable and looked at the break up of the middle layer into drops. \cite{Beeson-Jones/Woods:2015} analyzed general three-layer flows, and \cite{gin-daripa:hs-rad} performed a linear stability analysis for an arbitrary number of fluid layers in radial geometry.

To date, there are no known studies of multi-layer radial flows with {\bf variable} viscosity. The development of the theory for radial flow lags behind that of rectilinear flow because of the challenges due to the time-dependence of the problem. In particular, the curvature of the interfaces, the length of the middle layer(s) of fluid, and the spatially dependent viscous profile are all time-dependent. Previous stability studies of radial Hele-Shaw flows \citep{Anjos/Dias/Miranda:2015,Beeson-Jones/Woods:2015,Beeson-Jones/Woods:2017,Cardoso/Woods:1995,Dallaston/McCue:2013,Kim/Funada/etc:2009,Paterson:1981} allow for a time-dependent growth rate of instabilities, typically making use of a quasi-steady-state approximation. Despite this approximation, the stability results have often shown agreement with both physical and numerical experiments. In what follows, we keep with this tradition and use a quasi-steady-state approximation to derive an eigenvalue problem that has time-dependent coefficients and thus time-dependent eigenvalues. Additionally, the eigenvalue problem has eigenvalue-dependent boundary conditions which make the study of the eigenvalue problem mathematically interesting and challenging. The eigenvalues depend on many different parameters including the viscous profiles of the fluid layers, the interfacial tension at each interface, the curvature of the interfaces, and the fluid injection rate. Therefore, a numerical exploration of the vast parameter space is necessary. It is also important for applications in chemical EOR to investigate which viscous profiles minimize the instability for a given set of parameters, as is done in by \cite{daripa:tipm2012} in the rectilinear case.

It is worth pointing out another significant motivation behind this study. The simplest model to study incompressible porous media flow is to use the Hele-Shaw model which is based on the incompressibility condition and the Darcy's law in each of the phases. This model does not allow any mixing between the oil and water phases macroscopically and maintains sharp interfaces between the phases. Another model, called the Buckley-Leverett model, builds on this Hele-Shaw model by adding a saturation equation, a nonlinear hyperbolic conservation law with non-convex flux function which allows mixing between the phases (macroscopically) due to rarefaction waves behind the leading saturation front sweeping the oil ahead \citep{dglm88:polymer}. The rarefaction waves creates a viscous profile behind the front with viscosity gradually increasing towards the moving front. This mixing region with a viscous profile is usually finite in length which grows with time. The study of the stability of such composite solutions to such conservation laws is relevant for porous media but much too difficult. An appropriate model for this problem is the one under study which models the effect of rarefaction waves with a viscous profile between two interfaces and the shock front with a material interface having appropriate viscous jump at the interface. Such a study was partially carried out by \cite{daripa08:multi-layer} in the rectilinear geometry but the present study in the radial geometry is much harder as mentioned above and is also more relevant for porous media flow.

The paper is laid out as follows. In section \ref{sec:prelim} we perform a linear stability analysis of a point source driven three-layer {\bf radial} Hele-Shaw flow in which fluid between the two interfaces has a variable viscous profile. We use a time-dependent coordinate transformation to freeze the basic motion of the two interfaces in this new coordinate system and derive the associated eigenvalue problem in this new coordinate system. The growth rate of disturbances in the transformed coordinate system is related to the physical growth of disturbances of the interfaces in section \ref{sec:gr}. Section \ref{sec:constant} gives the restriction of the problem to the case of constant viscosity. In section \ref{sec:UpperBounds}, upper bounds on the growth rate are derived using the variational form of the problem. The nature of the eigenvalues and the completeness of the eigenfunctions are investigated in section \ref{sec:Char}. Numerical evaluation of the eigenvalues and their dependence on certain physical parameters are given in section \ref{sec:numerical}, and then we conclude in section \ref{sec:conclusions}.

\section{Preliminaries}\label{sec:prelim}
We consider a radial Hele-Shaw flow consisting of three regions of incompressible, immiscible fluid. By averaging across the gap, we may consider a two-dimensional flow domain in polar coordinates, $\Omega := (r,\theta) = \mathbb{R}^{2}$. The least viscous fluid with constant viscosity $\mu_i$ is injected into the center of the cell at a constant injection rate, $Q$. The most viscous fluid, with constant viscosity $\mu_o$, is the outermost fluid. The middle layer fluid has a smooth, axisymmetric viscous profile $\mu(r)$ where $\mu_i < \mu(r) < \mu_o$.
The fluid flow is governed by the following equations
\begin{equation}\label{VariableViscosity:eq:main}
\div{\bf u} = 0,\qquad \grad\;p=-\mu\;{\bf u}, \qquad \frac{\partial \mu}{\partial t} + \mathbf{u} \cdot \nabla \mu = 0,  \qquad \text{for } r \neq 0.
\end{equation}
The first equation $\eqref{VariableViscosity:eq:main}_1$ is the continuity equation for incompressible flow, the second equation $\eqref{VariableViscosity:eq:main}_2$ is Darcy's Law, and the third equation 
$\eqref{VariableViscosity:eq:main}_3$ is an advection equation for viscosity. Note that Darcy's Law for Hele-Shaw flows contains a permeability term $K = b^2/12$, but for the sake of simplicity here we have scaled viscosity by this term. Therefore, in what follows $\mu$ denotes the modified viscosity. We start with the fluids separated by circular interfaces with radii $R_1(0)$ and $R_2(0)$, where $R_1(t)$ and $R_2(t)$ are the positions of the interfaces at time $t$. This set-up is 
shown in Figure \ref{VariableViscosity:3Layer_Initial}.
\begin{figure}[!ht]
\centering
\includegraphics[width=.7\textwidth]{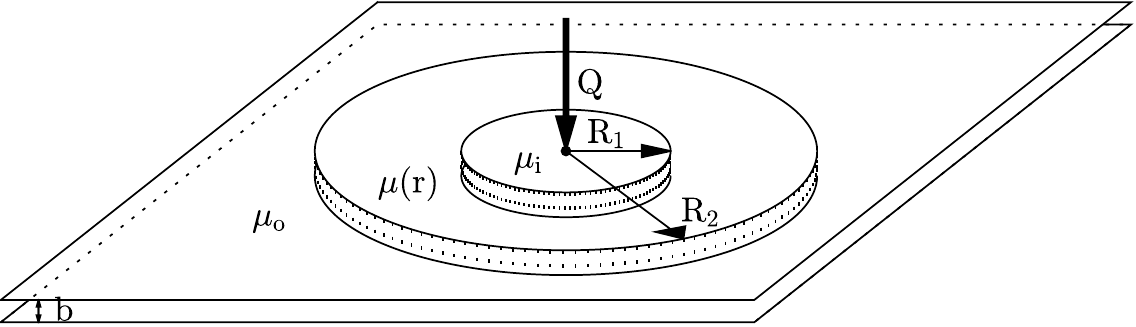}
\caption{The basic solution for three-layer flow}
\label{VariableViscosity:3Layer_Initial}
\end{figure}

The equations admit a simple basic solution in which all of the fluid moves outward radially with velocity $\mathbf{u} := \left(u_r,u_{\theta}\right) = \left(Q/(2 \pi r),0\right)$. The 
interfaces remain circular and their radii are given by $R_1(t) = \sqrt{Qt/\pi + R_1(0)^2}$ and $R_2(t) = \sqrt{Qt/\pi + R_2(0)^2}$. The pressure, $p_b = p_b(r)$, may be obtained by integrating equation $\eqref{VariableViscosity:eq:main}_2$.

We define the quantity $R_0(t) = \sqrt{Qt/\pi}$ and define the following coordinate transformation:
\begin{align}
 &\zeta = \frac{r^2 - R_0^2(t)}{R_2^2(t) - R_0^2(t)} = \frac{r^2 - R_0^2(t)}{R_2^2(0)}, \label{VariableViscosity:zeta} \\
 &\alpha = \theta, \label{VariableViscosity:alpha} \\
 &\tau = t. \label{VariableViscosity:tau}
\end{align}
The basic solution in these coordinates is $(u_{\zeta},u_{\alpha}) = (0,0)$ with the interfaces stationary at $\zeta = \zeta_1 := R_1^2(0)/R_2^2(0)$ and $\zeta = 1$.
 $\mu = \mu(\zeta)$ is now independent of time. 

We perturb this basic solution $\left(u_{\zeta} = 0, u_{\alpha} = 0, p_b, \mu \right)$ by $\left(\tilde{u}_{\zeta},\tilde{u}_{\alpha},\tilde{p}, \tilde{\mu}\right)$ where the disturbances are assumed to be small.
The linearized equations which govern these disturbances are
\begin{equation}\label{VariableViscosity:lindist}
\left.
 \begin{array}{l}
  \frac{\partial \tilde{u}_{\zeta}}{\partial \zeta} + \frac{1}{\zeta} \frac{\partial \tilde{u}_{\alpha}}{\partial \alpha} = 0 \\
  \frac{\partial \tilde{p}}{\partial \zeta} = -\frac{R_2^4(0)}{4(\zeta R_2^2(0) + R_0^2(\tau))} \mu \tilde{u}_{\zeta}  - \frac{QR_2^2(0)}{4 \pi (\zeta R_2^2(0) + R_0^2(\tau))} \tilde{\mu} \\
  \frac{\partial \tilde{p}}{\partial \alpha} = - \frac{\zeta R_2^2(0) + R_0^2(\tau)}{\zeta} \mu \tilde{u}_{\alpha} \\
  \frac{\partial \tilde{\mu}}{\partial \tau} + \tilde{u}_{\zeta} \frac{\partial \mu}{\partial \zeta}  = 0.
 \end{array} 
 \right\}
\end{equation}

We use separation of variables and decompose the disturbances into Fourier modes in the $\alpha$ coordinate so the disturbances are of the form
\begin{equation}\label{VariableViscosity:normalmodes}
 \left(\tilde{u}_{\zeta},\tilde{u}_{\alpha},\tilde{p}, \tilde{\mu}\right) = \big(f(\zeta,\tau), \delta(\zeta,\tau), \psi(\zeta,\tau),\phi(\zeta,\tau)\big) e^{in\alpha}. 
\end{equation}
Using this ansatz in equation \eqref{VariableViscosity:lindist} yields the following relations:
\begin{align}
&\frac{\partial \phi(\zeta,\tau)}{\partial \tau} = - \frac{d \mu}{d \zeta} f(\zeta,\tau), \label{VariableViscosity:modes3} \\
\begin{split}
&\frac{\partial}{\partial \zeta} \left\{ \left( \zeta R_2^2(0) + R_0^2(\tau) \right) \mu \frac{\partial f(\zeta,\tau)}{\partial \zeta} \right\} -\frac{n^2R_2^4(0)}{4(\zeta R_2^2(0) + R_0^2(\tau))} \mu f(\zeta,\tau)  \\ = &\frac{Qn^2R_2^2(0)}{4 \pi (\zeta R_2^2(0) + R_0^2(\tau))} \phi(\zeta,\tau). 
\end{split} \label{eigderiv}
\end{align}
In the innermost and outermost layers, the viscosity is constant and therefore $\phi(\zeta,\tau) \equiv 0$. In these regions, the solution of \eqref{eigderiv} is of the form
\begin{equation}\label{VariableViscosity:constGenSoln}
 f(\zeta,\tau) = \widetilde{C_1} \left(\zeta R_2^2(0) + R_0^2(\tau)\right)^{\frac{n}{2}} + \widetilde{C_2} \left(\zeta R_2^2(0) + R_0^2(\tau)\right)^{-\frac{n}{2}}.
\end{equation}

\subsection{Interface Conditions}
Recall that the inner interface in the $\zeta$-coordinate system is located at $\zeta = \zeta_1 := R_1^2(0)/R_2^2(0)$.
Let the disturbance of this interface be of the form $\eta_i = C_n(\tau) e^{in\alpha}$. The linearized kinematic condition at the inner interface is given by
\begin{equation}\label{VarInnerKin}
 C_n'(\tau) = f(\zeta_1,\tau).
\end{equation}
The outer interface is located at $\zeta = 1$. If the disturbance of the interface is of the form $\eta_o = D_n(\tau) e^{in\alpha}$, then the linearized kinematic condition is
\begin{equation}\label{VarOuterKin}
 D_n'(\tau) = f(1,\tau).
\end{equation}

The linearized dynamic interface condition at the inner interface is given by
\begin{equation*}
 \frac{2 R_1^2(\tau)}{R_2^2(0)} \Big( \tilde{p}^+(\zeta_1) - \tilde{p}^-(\zeta_1) \Big) - \eta_i \frac{Q}{2 \pi} \Big(\mu(\zeta_1) - \mu_i \Big)  =  T_1 \left( \frac{\eta_i + \frac{\partial^2 \eta_i}{\partial \alpha^2}}{R_1(\tau)}\right),
\end{equation*}
where $T_1$ is the interfacial tension.
Using the ansatz \eqref{VariableViscosity:normalmodes} and the system \eqref{VariableViscosity:lindist},
\begin{equation}\label{VariableViscosity:innerBC1}
 \frac{2 R_1^2(\tau)}{R_2^2(0)} \Big(\mu_i (f^-)'(\zeta_1,\tau) - \mu(\zeta_1) (f^+)'(\zeta_1,\tau) \Big) = \left\{ \frac{Qn^2}{2 \pi R_1^2(\tau)} \Big(\mu(\zeta_1) - \mu_i \Big)  -  T_1 \frac{n^4 - n^2}{R_1^3(\tau)}\right\} C_n(\tau).
\end{equation}
When $\zeta < \zeta_1$, $f$ is of the form given by \eqref{VariableViscosity:constGenSoln}. When $\tau = 0$, in order to avoid 
a singularity when $\zeta \to 0$, $f$ must be of the form $f(\zeta,\tau) = \widetilde{C_1} (\zeta R_2^2(0)+ R_0^2(\tau))^{\frac{n}{2}}.$
We assume this also to be true for $\tau > 0$. Then
\begin{equation}\label{VariableViscosity:innerBC2}
 (f^-)'(\zeta_1,\tau) = \frac{n R_2^2(0)}{2 R_1^2(\tau)} f(\zeta_1,\tau).
\end{equation}
Using \eqref{VariableViscosity:innerBC2} in \eqref{VariableViscosity:innerBC1},
\begin{equation}\label{VariableViscosity:innerBC3}
  n \mu_i f(\zeta_1,\tau) - \frac{2 R_1^2(\tau)}{R_2^2(0)} \mu(\zeta_1) (f^+)'(\zeta_1,\tau) = \left\{ \frac{Qn^2}{2 \pi R_1^2(\tau)} \Big(\mu(\zeta_1) - \mu_i \Big)  -  T_1 \frac{n^4 - n^2}{R_1^3(\tau)}\right\} C_n(\tau).
\end{equation}
Combining this with the kinematic interface condition \eqref{VarInnerKin},
\begin{equation}\label{VarInnerBC}
  C_n'(\tau) = \frac{2 R_1^2(\tau)}{n R_2^2(0)} \frac{\mu(\zeta_1)}{\mu_i} (f^+)'(\zeta_1,\tau) + \frac{F_1}{\mu_i}  C_n(\tau),
\end{equation}
where $F_1$ is given by
\begin{equation}\label{F1}
 F_1 = \frac{Qn}{2 \pi R_1^2(\tau)} \Big(\mu(\zeta_1) - \mu_i \Big)  -  T_1 \frac{n^3 - n}{R_1^3(\tau)}.
\end{equation}
A similar procedure for the outer interface yields the interface condition
\begin{equation}\label{VarOuterBC}
  D_n'(\tau) = - \frac{2 R_2^2(\tau)}{n R_2^2(0)} \frac{\mu(1)}{\mu_o} (f^-)'(1,\tau) + \frac{F_2}{\mu_o} D_n(\tau),
\end{equation}
where $F_2$ is given by
\begin{equation}\label{F2}
 F_2 = \frac{Qn}{2 \pi R_2^2(\tau)} \Big(\mu_o - \mu(1) \Big)  -  T_2 \frac{n^3 - n}{R_2^3(\tau)},
\end{equation}
and $T_2$ is the interfacial tension at the outer interface.

\subsection{Eigenvalue problem}
To this point we have the following system of equations where the field equations hold in the domain $(\zeta_1, 1)$:
\begin{equation}\label{system1}
\left.
\begin{array}{l l}
 \frac{\partial \phi(\zeta,\tau)}{\partial \tau} = - \frac{d \mu}{d \zeta} f(\zeta,\tau), \\
\frac{\partial}{\partial \zeta} \left\{ \left( \zeta R_2^2(0) + R_0^2(\tau) \right) \mu \frac{\partial f(\zeta,\tau)}{\partial \zeta} \right\} -\frac{n^2R_2^4(0)}{4(\zeta R_2^2(0) + R_0^2(\tau))} \mu f(\zeta,\tau)  = \frac{Qn^2R_2^2(0)\phi(\zeta,\tau)}{4 \pi (\zeta R_2^2(0) + R_0^2(\tau))} , & \\ 
  C_n'(\tau) = \frac{2 R_1^2(\tau)}{n R_2^2(0)} \frac{\mu(\zeta_1)}{\mu_i} f'(\zeta_1) + \frac{F_1}{\mu_i}  C_n(\tau), & \\
  D_n'(\tau) = - \frac{2 R_2^2(\tau)}{n R_2^2(0)} \frac{\mu(1)}{\mu_o} f'(1) + \frac{F_2}{\mu_o} D_n(\tau), &
\end{array} \right\}
\end{equation}
where we have dropped the superscripts ``+'' and ``-''. Using a quasi-steady-state approximation (QSSA) in which $\tau$ (and hence $R(\tau)$) are frozen so that the functions $\phi(x,\tau)$, $C_n(\tau)$, and $D_n(\tau)$ experience short-time exponential growth satisfying
\begin{equation}\label{growth}
\left.
\begin{array}{l}
\frac{\partial \phi(\zeta,\tau)}{\partial \tau} = \sigma(\tau) \phi(\zeta,\tau), \\
C_n'(\tau) = \sigma(\tau) C_n(\tau), \\
D_n'(\tau) = \sigma(\tau) D_n(\tau),
\end{array} \right\}
\end{equation}
for some growth rate $\sigma(\tau)$. Plugging \eqref{growth} into \eqref{system1} and using \eqref{VarInnerKin} and \eqref{VarOuterKin}, $(f,\sigma)$ is a solution to the following eigenvalue problem in the domain $(\zeta_1, 1)$:
\begin{equation}\label{Eigproblem}
\left.
\begin{array}{l l}
 \Big(\left(\zeta R_2^2(0) + R_0^2(\tau) \right) \mu f'(\zeta) \Big)' - \frac{n^2 R_2^4(0)}{4(\zeta R_2^2(0) + R_0^2(\tau))} \mu f(\zeta) = - \frac{Q n^2 R_2^2(0)}{4 \pi (\zeta R_2^2(0) + R_0^2(\tau))} \frac{1}{\sigma(\tau)} \frac{d \mu}{d \zeta} f(\zeta), \\ 
  \frac{2 R_1^2(\tau)}{n R_2^2(0)} \mu(\zeta_1) f'(\zeta_1) = \left(\mu_i - \frac{F_1}{\sigma(\tau)} \right) f(\zeta_1), & \\
 -\frac{2 R_2^2(\tau)}{n R_2^2(0)} \mu(1) f'(1) = \left(\mu_o - \frac{F_2}{\sigma(\tau)} \right) f(1). &
\end{array} \right\}
\end{equation}
The eigenvalues of system \eqref{Eigproblem} are the time-dependent growth rates of the disturbances of the system. The QSSA allows for the considerable analysis and computation of growth rates that follows.

\section{Relating the growth of interfacial disturbances in the $\zeta$-coordinates with the physical coordinates}\label{sec:gr}
We now relate the growth of the interfacial disturbances in the $\zeta$-coordinates to the same in the physical coordinate system. 
We start with the inner interface. Recall that in the transformed coordinates, the inner interface was disturbed by 
$C_n(\tau) e^{in\alpha}$. Therefore, it is located at $\zeta = \zeta_1 + C_n(\tau) e^{in\alpha}$. Thus, the position of the interface in the physical coordinates is 
\begin{equation*}
 r = \sqrt{\zeta R_2^2(0) + R_0^2(t)} = \sqrt{(\zeta_1 + C_n(\tau) e^{in\alpha}) R_2^2(0) + R_0^2(t)}.
\end{equation*}
Expanding about $\zeta = \zeta_1$,
\begin{equation*}
 r = R_1(\tau) + \frac{R_2^2(0)}{2 R_1(\tau)} C_n(\tau) e^{in\alpha} + \mathcal{O}(C_n^2(\tau)).
\end{equation*}
If we write the disturbance in the physical coordinates as $A_n(t) e^{in\theta}$ (that is, the interface is located at $r = R_1(t) + A_n(t) e^{in\theta}$), then, within linear approximation,
\begin{equation}\label{InterfaceCompare}
 A_n(t) = \frac{R_2^2(0)}{2 R_1(\tau)} C_n(\tau).
\end{equation}
This implies that
\begin{equation}\label{Compare_inner}
 \frac{A_n'(t)}{A_n(t)} =  \frac{C_n'(\tau)}{C_n(\tau)} -\frac{Q}{2\pi R_1^2(t)}.
\end{equation}
Following the same process, the growth of the outer interface is 
\begin{equation}\label{Compare_outer}
 \frac{B_n'(t)}{B_n(t)} =  \frac{D_n'(\tau)}{D_n(\tau)} -\frac{Q}{2\pi R_2^2(t)},
\end{equation}
where the outer interface is located at $r = R_2(t) + B_n(t) e^{in\theta}$.

\section{Constant Viscosity Fluids}\label{sec:constant}
We now consider the case in which all of the fluids have constant viscosity, first for two-layer flows and then for three-layer. We do this to demonstrate that the variable viscosity formulation can recover previous results in the constant viscosity limit. Through this process, we also present some new results on three-layer constant viscosity flows.

When there are only two fluids (i.e. one interface located at $r = R(t)$), the above analysis holds with the coordinate transformation
\begin{equation*}
 \zeta = \frac{r^2 - R_0^2(t)}{R^2(0)}.
\end{equation*}
In the new coordinates, the basic solution has the interface fixed at $\zeta = 1$.
Let $\mu_i$ denote the viscosity of the inner fluid and $\mu_o$ denote the viscosity of the outer fluid. Analogous to equations \eqref{VarOuterKin} and \eqref{VariableViscosity:innerBC1}, the kinematic interface condition is
\begin{equation}\label{VariableViscosity:twolayer_kinBC}
 D_n'(\tau) = f(1,\tau),
\end{equation}
and the dynamic interface condition is
\begin{equation}\label{VariableViscosity:twolayerdynBC}
 \frac{2 R^2(\tau)}{R^2(0)} \Big(-\mu_o (f^+)'(1,\tau) + \mu_i (f^-)'(1,\tau) \Big) = \left\{ \frac{Qn^2}{2 \pi R^2(\tau)} \Big(\mu_o - \mu_i \Big)  -  T \frac{n^4 - n^2}{R^3(\tau)}\right\} D_n(\tau),
\end{equation}
where $T$ is the interfacial tension and $D_n(\tau)$ is the amplitude of the disturbance of the interface with wave number $n$.
Also, as stated in the derivation of the interface conditions above,
\begin{equation*}
 f(\zeta,\tau) = \widetilde{C_1} \Big(\zeta R^2(0) + R_0^2(\tau)\Big)^{\frac{n}{2}}, \qquad \zeta < 1,
\end{equation*}
and
\begin{equation*}
 f(\zeta,\tau) = \widetilde{C_2} \Big(\zeta R^2(0) + R_0^2(\tau)\Big)^{-\frac{n}{2}}, \qquad \zeta > 1.
\end{equation*}
Using these in equations \eqref{VariableViscosity:twolayer_kinBC} and \eqref{VariableViscosity:twolayerdynBC} gives the two-layer growth rate
\begin{equation}\label{VariableViscosity:2layerGR}
 \sigma := \frac{D_n'(\tau)}{D_n(\tau)}   = \frac{Qn}{2 \pi R^2(\tau)} \frac{\mu_o - \mu_i}{\mu_o + \mu_i} - \frac{T}{\mu_i + \mu_o} \frac{n^3 - n}{R^3(\tau)}.
\end{equation}
This is an expression for the growth rate of the disturbance of the interface in the $\zeta$-coordinate system. This problem can be solved in the original $r$-coordinate system, 
and the result is a classic one \citep{Paterson:1981}. We recall this result, which has been reproduced using our current notation by \cite{gin-daripa:hs-rad}. If $B_n(t)$ is the amplitude of the disturbance with wave number $n$, the growth rate is
\begin{equation}\label{VariableViscosity:2layer4}
  \frac{B_n'(t)}{B_n(t)}  = \frac{Qn}{2 \pi R^2(t)} \frac{\mu_o - \mu_i}{\mu_o + \mu_i} - \frac{T}{\mu_o + \mu_i} \frac{n^3 - n}{R^3(t)} - \frac{Q}{2 \pi R^2(t)}.
\end{equation}
The relationship between equations \eqref{VariableViscosity:2layerGR} and \eqref{VariableViscosity:2layer4} is consistent with the comparison of the growth rates in the two different coordinate systems given by equation \eqref{Compare_outer}.

We now turn to three-layer flows in which the fluid in the middle layer also has constant viscosity, $\mu_1$. This situation has been investigated by \cite{Beeson-Jones/Woods:2015} and \cite{gin-daripa:constant_interfaces}, and it has been found that the magnitudes of interfacial disturbances $A_n(t)$ and $B_n(t)$ are governed by the following system of ODE's
\begin{equation}\label{ODE}
 \frac{d}{dt}
\begin{pmatrix}
 A_n(t) \\
 B_n(t) 
\end{pmatrix} = \mathbf{M}_1^r(t) 
\begin{pmatrix}
 A_n(t) \\
 B_n(t) 
\end{pmatrix},
\end{equation}
where $\mathbf{M}_1^r(t)$ is the $2 \times 2$ matrix with entries given by
\begin{equation}\label{M_3Layer}
\begin{split}
 \Big(\mathbf{M}_1^r(t)\Big)_{11} &=  \frac{\left\{(\mu_o + \mu_1) - (\mu_o - \mu_1) \left(\frac{R_1}{R_2}\right)^{2} \right\}  F_1 }{(\mu_1 - \mu_i) (\mu_o - \mu_1) \left(\frac{R_1}{R_2}\right)^{2} + (\mu_1 + \mu_i) (\mu_o + \mu_1)} - \frac{Q}{2 \pi R_1^2}, \\
 \Big(\mathbf{M}_1^r(t)\Big)_{12} &=  \frac{2 \mu_1 \left(\frac{R_1}{R_2}\right)^{n-1} F_2}{(\mu_1 - \mu_i) (\mu_o - \mu_1) \left(\frac{R_1}{R_2}\right)^{2} + (\mu_1 + \mu_i) (\mu_o + \mu_1)}, \\
 \Big(\mathbf{M}_1^r(t)\Big)_{21} &= \frac{2 \mu_1 \left(\frac{R_1}{R_2}\right)^{n+1} F_1}{(\mu_1 - \mu_i) (\mu_o - \mu_1) \left(\frac{R_1}{R_2}\right)^{2} + (\mu_1 + \mu_i) (\mu_o + \mu_1)}, \\
 \Big(\mathbf{M}_1^r(t)\Big)_{22} &=  \frac{\left\{(\mu_1 + \mu_i) + (\mu_1 - \mu_i) \left(\frac{R_1}{R_2}\right)^{2} \right\}  F_2 }{(\mu_1 - \mu_i) (\mu_o - \mu_1) \left(\frac{R_1}{R_2}\right)^{2} + (\mu_1 + \mu_i) (\mu_o + \mu_1)} - \frac{Q}{2 \pi R_2^2}.
\end{split}
\end{equation}
We recall equation \eqref{InterfaceCompare} which compares the interfacial disturbance of the inner interface in the $r$-coordinates and the $\zeta$-coordinates, and also consider the corresponding equation for the outer interface:
\begin{equation}\label{InterfaceCompare2}
A_n(t) = \frac{R_2^2(0)}{2 R_1(\tau)} C_n(\tau), \qquad B_n(t) = \frac{R_2^2(0)}{2 R_2(\tau)} D_n(\tau).
\end{equation}
Equations \eqref{ODE} and \eqref{InterfaceCompare2} gives the matrix equation
\begin{equation}\label{Mr_to_Mzeta1}
  \frac{d}{dt}
\begin{pmatrix}
 A_n(t) \\
 B_n(t) 
\end{pmatrix} =
\mathbf{M}_1^r(t) \frac{R_2^2(0)}{2} \mathbf{R}^{-1}
\begin{pmatrix}
 C_n(\tau) \\
 D_n(\tau) 
\end{pmatrix},
\end{equation}
where
\begin{equation}
 \mathbf{R} = \begin{pmatrix}
 R_1 & 0 \\
0 & R_2 
\end{pmatrix}.
\end{equation}
Taking derivatives of \eqref{InterfaceCompare2} and rewriting the resulting equations in matrix form, we obtain
\begin{equation}\label{Mr_to_Mzeta2}
  \frac{d}{dt}
\begin{pmatrix}
 A_n(t) \\
 B_n(t) 
\end{pmatrix} = \frac{R_2^2(0)}{2} \mathbf{R}^{-1} \frac{d}{d \tau}
\begin{pmatrix}
 C_n(\tau) \\
 D_n(\tau) 
\end{pmatrix} - \frac{R_2^2(0)}{2} \mathbf{R}^{-1} \mathbf{Q}
\begin{pmatrix}
 C_n(\tau) \\
 D_n(\tau) 
\end{pmatrix},
\end{equation}
where
\begin{equation}
 \mathbf{Q} = \begin{pmatrix}
 \frac{Q}{2 \pi R_1^2} & 0 \\
0 & \frac{Q}{2 \pi R_2^2}
\end{pmatrix}.
\end{equation}
Combining \eqref{Mr_to_Mzeta1} and \eqref{Mr_to_Mzeta2},
\begin{equation}\label{Mr_to_Mzeta}
  \frac{d}{d\tau}
\begin{pmatrix}
 C_n(\tau) \\
 D_n(\tau) 
\end{pmatrix} = \mathbf{M}_1^{\zeta}(\tau)
\begin{pmatrix}
 C_n(\tau) \\
 D_n(\tau) 
\end{pmatrix},
\end{equation}
where $\mathbf{M}_1^{\zeta} = \mathbf{R} \mathbf{M}_1^r \mathbf{R}^{-1} + \mathbf{Q}$. The entries of $\mathbf{M}_1^{\zeta}(\tau)$ are
\begin{equation}\label{M_zeta}
\begin{split}
 \Big(\mathbf{M}_1^{\zeta}(\tau)\Big)_{11} &=  \frac{\left\{(\mu_o + \mu_1) - (\mu_o - \mu_1) \left(\frac{R_1}{R_2}\right)^{2} \right\}  F_1 }{(\mu_1 - \mu_i) (\mu_o - \mu_1) \left(\frac{R_1}{R_2}\right)^{2} + (\mu_1 + \mu_i) (\mu_o + \mu_1)}, \\
\Big(\mathbf{M}_1^{\zeta}(\tau)\Big)_{12} &=  \frac{2 \mu_1 \left(\frac{R_1}{R_2}\right)^{n} F_2}{(\mu_1 - \mu_i) (\mu_o - \mu_1) \left(\frac{R_1}{R_2}\right)^{2} + (\mu_1 + \mu_i) (\mu_o + \mu_1)}, \\
 \Big(\mathbf{M}_1^{\zeta}(\tau)\Big)_{21} &= \frac{2 \mu_1 \left(\frac{R_1}{R_2}\right)^{n} F_1}{(\mu_1 - \mu_i) (\mu_o - \mu_1) \left(\frac{R_1}{R_2}\right)^{2} + (\mu_1 + \mu_i) (\mu_o + \mu_1)}, \\
 \Big(\mathbf{M}_1^{\zeta}(\tau)\Big)_{22} &=  \frac{\left\{(\mu_1 + \mu_i) + (\mu_1 - \mu_i) \left(\frac{R_1}{R_2}\right)^{2} \right\}  F_2 }{(\mu_1 - \mu_i) (\mu_o - \mu_1) \left(\frac{R_1}{R_2}\right)^{2} + (\mu_1 + \mu_i) (\mu_o + \mu_1)}.
\end{split}
\end{equation}
There are several important things to note from the relationship between the matrices $\mathbf{M}_1^{\zeta}$ and $\mathbf{M}_1^r$:

\begin{enumerate}
 
\item It was demonstrated by \cite{gin-daripa:constant_interfaces} that $\mathbf{M}_1^r$ can have complex eigenvalues.  However, it is shown in the next section (section \ref{sec:UpperBounds}) that the problem in the $\zeta$-coordinates has real growth rates. This analysis holds even for constant viscosity. Therefore, $\mathbf{M}_1^{\zeta}$ has real eigenvalues.

\item  $\mathbf{M}_1^{\zeta}$ can also be expressed as $\mathbf{M}_1^{\zeta} = \mathbf{R} \left(\mathbf{M}_1^r  + \mathbf{Q} \right) \mathbf{R}^{-1}$. Since a similarity transformation does not change eigenvalues, the eigenvalues of $\mathbf{M}_1^{\zeta}$ are the eigenvalues of $\mathbf{M}_1^r + \mathbf{Q}$ where $\mathbf{Q}$ is a diagonal matrix. Thus it is this diagonal matrix $\mathbf{Q}$ which, when added to $\mathbf{M}_1^r$, converts the complex eigenvalues to real and leaves real eigenvalues real.

\item Both $\mathbf{M}_1^r$ and $\mathbf{M}_1^{\zeta}$ have real eigenvalues when $F_1$ and $F_2$ defined respectively in \eqref{F1} and \eqref{F2} have the same sign. Define the matrices:
\begin{equation}\label{C}
\mathbf{E} = 
\begin{pmatrix}
R_1 \sqrt{|F_1|} & 0 \\
 0 & R_2 \sqrt{|F_2|} 
\end{pmatrix}, \qquad
\mathbf{F} = 
\begin{pmatrix}
\sqrt{|F_1|} & 0 \\
 0 & \sqrt{|F_2|} 
\end{pmatrix}
\end{equation}
then $\mathbf{E} \mathbf{M}_1^r \mathbf{E}^{-1}$ and $\mathbf{F} \mathbf{M}_1^{\zeta} \mathbf{F}^{-1}$ are real symmetric matrices. Therefore, they are similar to real symmetric (i.e. self-adjoint) matrices and have real eigenvalues. However, the argument breaks down when $F_1$ and $F_2$ have opposite signs because $\mathbf{E} \mathbf{M}_1^r \mathbf{E}^{-1}$ and $\mathbf{F} \mathbf{M}_1^{\zeta} \mathbf{F}^{-1}$ are not symmetric. This shows that the complex eigenvalues of $\mathbf{M}_1^r$ can only occur when $F_1 F_2 \leq 0$.
\end{enumerate}

\section{Upper Bounds}\label{sec:UpperBounds}
To derive an upper bound on the growth rate, we take an inner product of $\eqref{Eigproblem}_1$ with $f$.
Using integration by parts along with the boundary conditions $\eqref{Eigproblem}_2$ and $\eqref{Eigproblem}_3$ and solving for $\sigma$ yields
\begin{equation}\label{VariableViscosity:RayQuotient}
 \sigma = \frac{n F_1 |f(\zeta_1)|^2 + n F_2 |f(1)|^2 + \frac{Qn^2}{2 \pi} I_1}{ n \mu_i  |f(\zeta_1)|^2 + n \mu_o |f(1)|^2 + \frac{2}{R_2^2(0)} I_2 +\frac{n^2 R_2^2(0)}{2} I_3},
\end{equation}
where
\begin{align}
 I_1 &= \int_{\zeta_1}^1 \frac{\mu'(\zeta)}{\zeta R_2^2(0) + R_0^2(\tau)} |f(\zeta)|^2 d\zeta, \label{VariableViscosity:I1} \\
 I_2 &= \int_{\zeta_1}^{1} \left(\zeta R_2^2(0) + R_0^2(\tau) \right) \mu(\zeta) |f'(\zeta)|^2 d\zeta, \label{VariableViscosity:I2}\\
 I_3 &= \int_{\zeta_1}^1 \frac{\mu(\zeta)}{\zeta R_2^2(0) + R_0^2(\tau)} |f(\zeta)|^2 d\zeta. \label{VariableViscosity:I3}
\end{align}

Note that all terms in \eqref{VariableViscosity:RayQuotient} are real. Therefore, \textbf{$\sigma$ is real for all wave numbers.} This is a product of the change of variables from the $r$-coordinates to the $\zeta$-coordinates. It is shown by \cite{gin-daripa:hs-rad} that the growth rate can be complex for constant viscosity flows in the $r$-coordinates. 

When $\sigma > 0$, we may ignore the 
positive term containing $I_2$ in the denominator and get
\begin{equation*}
 \sigma < \frac{n F_1 |f(\zeta_1)|^2 + n F_2 |f(1)|^2 + \frac{Qn^2}{2 \pi} I_1}{ n \mu_i  |f(\zeta_1)|^2 + n \mu_o |f(1)|^2 +\frac{n^2 R_2^2(0)}{2} I_3}.
\end{equation*}
We use the following inequality
\begin{equation*}
 \frac{\sum\limits_{i=1}^{N} A_ix_i}{\sum\limits_{i=1}^{N} B_ix_i} \leq \max_i \left\{\frac{A_i}{B_i}\right\},
\end{equation*}
which holds for any $N$ if $A_i > 0$, $B_i > 0$, and $X_i > 0$ for all
$i = 1,...,N$. By using this inequality with $N = 3$,
\begin{equation*}
 \sigma < \max \left\{\frac{F_1}{\mu_i}, \frac{F_2}{\mu_o}, \frac{Q}{\pi R_2^2(0)} \frac{I_1}{I_3}  \right\}.
\end{equation*}
But
\begin{align*}
 \frac{I_1}{I_3} <  \frac{{\displaystyle \sup_{\zeta \in (\zeta_1,1)}} \mu'(\zeta) }{{\displaystyle\inf_{\zeta \in (\zeta_1,1)}} \mu(\zeta)}  < \frac{{\displaystyle\sup_{\zeta \in (\zeta_1,1)}} \mu'(\zeta)}{\mu_i}.
\end{align*}
Therefore,
\begin{equation}\label{VariableViscosity:UpperBound}
 \sigma < \max \left\{\frac{F_1}{\mu_i}, \frac{F_2}{\mu_o}, \frac{Q}{\pi R_2^2(0)} \frac{1}{\mu_i} \sup_{\zeta \in (\zeta_1,1)} \mu'(\zeta)  \right\}.
\end{equation}
Using the definitions of $F_1$ and $F_2$ given by \eqref{F1} and \eqref{F2}, 
\begin{equation}\label{VariableViscosity:UpperBound2}
\begin{split}
 \sigma < \max & \left\{ \frac{Qn}{2 \pi R_1^2(\tau)}\left(\frac{\mu(\zeta_1) - \mu_i}{\mu_i}\right) - \frac{T_1}{\mu_i} \frac{n^3-n}{R_1^3(\tau)}, \right. \\
&  \left. \frac{Qn}{2 \pi R_2^2(\tau)}\left(\frac{\mu_o - \mu(1)}{\mu_o}\right) - \frac{T_2}{\mu_o} \frac{n^3-n}{R_2^3(\tau)}, \frac{Q}{\pi R_2^2(0)} \frac{1}{\mu_i} \sup_{\zeta \in (\zeta_1,1)} \mu'(\zeta)  \right\},
\end{split}
\end{equation}
which is the modal upper bound for a wave with wave number $n$. We can find an absolute upper bound for all wave numbers by taking the maximum of the first two terms over all values of $n$. The absolute upper bound is
\begin{equation}\label{VariableViscosity:AbsUpperBound}
\begin{split}
 \sigma < \max & \left\{\frac{2 T_1}{\mu_i R_1^3(\tau)} \left(\frac{Q R_1(\tau)}{6 \pi T_1} (\mu(\zeta_1) - \mu_i) + \frac{1}{3}\right)^{\frac{3}{2}}, \right. \\ 
 & \left. \frac{2 T_2}{\mu_o R_2^3(\tau)} \left(\frac{Q R_2(\tau)}{6 \pi T_2} (\mu_o - \mu(1)) + \frac{1}{3}\right)^{\frac{3}{2}}, \frac{Q}{\pi R_2^2(0)} \frac{1}{\mu_i} \sup_{\zeta \in (\zeta_1,1)} \mu'(\zeta)  \right\}.
\end{split}
\end{equation}

\section{Characterization of the Eigenvalues and Eigenfunctions}\label{sec:Char}
Using $\lambda = 1/\sigma$, the eigenvalue problem \eqref{Eigproblem} can be written as
\begin{equation}\label{VariableViscosity:EigLambda}
\left.
\begin{array}{l l}
 \Big(\left(\zeta R_2^2(0) + R_0^2(\tau) \right) \mu f'(\zeta) \Big)' - \left( \frac{n^2 R_2^4(0)}{4(\zeta R_2^2(0) + R_0^2(\tau))} \mu  -\frac{Q n^2 R_2^2(0)}{4 \pi (\zeta R_2^2(0) + R_0^2(\tau))} \mu' \lambda  \right)  f(\zeta) =  0, \\ 
 \left(\mu_i - \lambda F_1 \right) f(\zeta_1)  - \frac{2 R_1^2(\tau)}{Rn_2^2(0)} \mu(\zeta_1) f'(\zeta_1) = 0, & \\
 \left(\mu_o - \lambda F_2 \right) f(1) + \frac{2 R_2^2(\tau)}{Rn_2^2(0)} \mu(1) f'(1) = 0. &
\end{array} \right\}
\end{equation}
Note that $F_1$ and $F_2$ are positive for small values of $n$ and negative for large values of $n$ (see equations \eqref{F1} and \eqref{F2}). From the upper bound \eqref{VariableViscosity:UpperBound}, we can see that as long as the viscous gradient $\mu'(\zeta)$ is not too large, the maximum value of $\sigma$ will occur when $F_1$ and $F_2$ are positive.
For this range of wave numbers, we have the following characterization of the eigenvalues and eigenfunctions.

\begin{theorem}\label{VariableViscosity:Thm1}
Let $F_1$, $F_2$, $Q$, $n$, $\mu_i$, $\mu_o > 0$. Let $\mu(\zeta)$ be a positive, strictly increasing function in $C^1([\zeta_1,1])$. 
Then the eigenvalue problem \eqref{VariableViscosity:EigLambda} has a countably infinite number of real eigenvalues that can be ordered
\begin{equation*}
  0 < \lambda_0 < \lambda_1 < \lambda_2 < ...
\end{equation*}
with the property that for the corresponding eigenfunctions, $\left\{f_i\right\}_{i=0}^{\infty}$, $f_i$ has exactly $i$ zeros in the interval $(\zeta_1,1)$. Additionally, the eigenfunctions are continuous with 
a continuous derivative.
\end{theorem}

\begin{proof}
 The fact that there are a countably infinite number of real eigenvalues that can be ordered and corresponding eigenfunctions with the prescribed number of zeros is proven by Ince \cite[p. 232-233]{Ince:1956} in Theorem I and Theorem II using 
 \begin{align*}
  &a = \zeta_1, \qquad
  b = 1, \qquad
  K(x,\lambda) = \left(x R_2^2(0) + R_0^2(\tau) \right)\mu(x), \\
  &G(x,\lambda) = \frac{n^2 R_2^4(0)}{4(x R_2^2(0) + R_0^2(\tau))} \mu(x)  -\frac{Q n^2 R_2^2(0)}{4 \pi (x R_2^2(0) + R_0^2(\tau))} \mu'(x) \lambda, \\
  &\alpha = \frac{2 R_1^2(\tau)}{Rn_2^2(0)} \mu(\zeta_1), \qquad
  \alpha' = \mu_i - \lambda F_1, \\
  &\beta = \frac{2 R_2^2(\tau)}{Rn_2^2(0)} \mu(1), \qquad
  \beta' = \mu_o - \lambda F_2.
 \end{align*}
 The regularity of the eigenfunctions comes from the existence theorem of Ince \cite[p. 73]{Ince:1956}. We saw from equation \eqref{VariableViscosity:RayQuotient} that $\sigma$ is real for all $n$, and a closer look at each term 
 in \eqref{VariableViscosity:RayQuotient} shows that if $F_1,F_2 > 0$ and $\mu(\zeta), \mu'(\zeta) >0$, then all terms are positive and $\sigma > 0$.
\end{proof}
\bigskip

\subsection{Self-Adjointness and Expansion Theorem}
We rewrite equation \eqref{VariableViscosity:EigLambda} as
\begin{equation}\label{VariableViscosity:EigWalter}
\left.
\begin{array}{l l}
 -\Big(\left(\zeta R_2^2(0) + R_0^2(\tau) \right) \mu f'(\zeta) \Big)' + \left( \frac{n^2 R_2^4(0)}{4(\zeta R_2^2(0) + R_0^2(\tau))} \right) \mu f(\zeta) =  \frac{Q n^2 R_2^2(0)}{4 \pi (\zeta R_2^2(0) + R_0^2(\tau))} \mu' \lambda f(\zeta), \\ 
   - \left(-\frac{\mu_i}{F_1} f(\zeta_1) + \frac{2 R_1^2(\tau)}{Rn_2^2(0) F_1} \mu(\zeta_1) f'(\zeta_1)\right)  = \lambda f(\zeta_1), & \\
 -\left(-\frac{\mu_o}{F_2} f(1) - \frac{2 R_2^2(\tau)}{Rn_2^2(0)F_2} \mu(1) f'(1)\right) = \lambda f(1). &
\end{array} \right\}
\end{equation}
This is of the form
\begin{equation}
\left.
\begin{array}{l l}
 Tf := \frac{1}{r} \left\{-\left(p f'\right)' + qf \right\} = \lambda f, \qquad &\zeta_1 < \zeta < 1, \\
   - \left(\beta_{11} f(\zeta_1) - \beta_{12} f'(\zeta_1) \right) = \lambda \left( \alpha_{11} f(\zeta_1) - \alpha_{12} f'(\zeta_1) \right) , & \\
- \left(\beta_{21} f(1) - \beta_{22} f'(1) \right) = \lambda \left( \alpha_{21} f(1) - \alpha_{22} f'(1) \right), &
\end{array} \right\}
\end{equation}
where
\begin{equation}\label{VariableViscosity:WalterNotation}
\begin{array}{l l}
 p(\zeta) = \left(\zeta R_2^2(0) + R_0^2(\tau) \right) \mu(\zeta), & \\ 
 q(\zeta) = \frac{n^2 R_2^4(0) \mu(\zeta)}{4(\zeta R_2^2(0) + R_0^2(\tau))}, & \\
 r(\zeta) = \frac{Q n^2 R_2^2(0) \mu'(\zeta)}{4 \pi (\zeta R_2^2(0) + R_0^2(\tau))}, & \\
 \beta_{11} = -\frac{\mu_i}{F_1}, & \beta_{12} = -\frac{2 R_1^2(\tau)}{Rn_2^2(0) F_1} \mu(\zeta_1), \\
 \alpha_{11} = 1, & \alpha_{12} = 0, \\
 \beta_{21} = -\frac{\mu_o}{F_2}, & \beta_{22} = \frac{2 R_2^2(\tau)}{Rn_2^2(0) F_2} \mu(1), \\
 \alpha_{21} = 1, & \alpha_{22} = 0.
\end{array}
\end{equation}
Given the same assumptions as in Theorem \ref{VariableViscosity:Thm1}, we have the following theorem from \cite{Walter:1973}.
\begin{theorem}\label{VariableViscosity:Thm2}
Let $F_1$, $F_2$, $Q$, $n$, $\mu_i$, $\mu_o > 0$. Let $\mu(\zeta)$ be a positive, strictly increasing function in $C^1([\zeta_1,1])$.
Let $p(\zeta)$, $q(\zeta)$, and $r(\zeta)$ be defined by \eqref{VariableViscosity:WalterNotation}. Let 
\begin{equation*}
 L^2_r(\zeta_1,1) = \left\{f(\zeta) \Big\vert \int_{\zeta_1}^{1} |f(\zeta)|^2 r(\zeta) d \zeta < \infty \right\},
\end{equation*}
and define the operator $T$ on $L^2_r(\zeta_1,1)$ by
\begin{equation*}
 Tf := \frac{1}{r} \left\{-\left(p f'\right)' + q f\right\}.
\end{equation*}
Define the measure:
\begin{equation}
 \nu(M) := \left\{
 \begin{array}{l l}
  \frac{nR_2^2(0)F_1}{2}, \qquad & \text{for } M = \{\zeta_1\} \\
  \int_{M} r(\zeta) d\zeta, \qquad & \text{for } M \subset (\zeta_1,1) \\
  \frac{nR_2^2(0)F_2}{2}, \qquad & \text{for } M = \{1\}. \\
 \end{array}
 \right.
\end{equation}
We consider the Hilbert space $H := L^2([\zeta_1,1];\nu)$. Consider the operator $A$ with domain
\begin{equation}
 D(A) = \{f \in H \vert f,f' \text{ absolutely continuous in } (\zeta_1,1), T \in L^2_r(\zeta_1,1) \},
\end{equation}
and defined by
\begin{equation}
 (Af)(\zeta) = \left\{
 \begin{array}{l l}
  \displaystyle{\lim_{\zeta \to \zeta_1}} \left(\frac{\mu_i}{F_1} f(\zeta) - \frac{2 R_1^2(\tau)}{Rn_2^2(0) F_1} \mu(\zeta_1) f'(\zeta)\right), \qquad & \text{if } \zeta = \{\zeta_1\} \\
  (Tf)(\zeta), \qquad & \text{if } \zeta \in (\zeta_1,1) \\
  \displaystyle{\lim_{\zeta \to 1}} \left(\frac{\mu_o}{F_2} f(\zeta) + \frac{2 R_2^2(\tau)}{Rn_2^2(0) F_2} \mu(1) f'(\zeta)\right), \qquad & \text{if } \zeta = \{1\}. \\
 \end{array}
 \right.
\end{equation}
Then $(f,\lambda)$ satisfies \eqref{VariableViscosity:EigWalter} if and only if $Af = \lambda f$. 
$A$ is a self-adjoint operator on $H$ and for any $u \in H$,
\begin{equation*}
 u = \sum_{k=0}^{\infty} f_k \int_{\zeta_1}^{1} u(\zeta) f_k(\zeta) d \nu,
\end{equation*}
where the $f_k$ are the eigenfunctions of $A$.
\end{theorem}

\subsection{Notes on the Assumptions}
The above theorem holds assuming 
\begin{equation*}
 p \in C^1([\zeta_1,1]), \qquad q \in C^0([\zeta_1,1]), \qquad r \in C^0([\zeta_1,1]),
\end{equation*}
and $p(\zeta) > 0, r(\zeta) > 0$ for $\zeta \in [\zeta_1,1]$. This is satisfied if
\begin{equation}
 \mu(\zeta) \in C^1([\zeta_1,1]), \qquad \mu'(\zeta) > 0.
\end{equation}
It is also assumed that $F_1, F_2 > 0$ which holds when
\begin{equation}
 n < \min \left\{\sqrt{\frac{Q R_1(\tau)}{2 \pi T_1} (\mu(\zeta_1) - \mu_i) + 1}, \sqrt{\frac{Q R_2(\tau)}{2 \pi T_2} (\mu_o - \mu(1)) + 1} \right\}.
\end{equation}
A different theory will be necessary to consider non-monotonic viscous profiles or large wave numbers.

\section{Numerical Results}\label{sec:numerical}
We now investigate the growth rate of disturbances by numerically computing the eigenvalues of the eigenvalue problem \eqref{VariableViscosity:EigLambda}. This eigenvalue problem has time-dependent coefficients and boundary conditions which depend on the eigenvalues. Thus the dispersion relation for this problem depends on time. The eigenvalues are computed using a pseudo-spectral Chebyshev method. The eigenvalues $\lambda$ are then inverted to find the growth rates $\sigma$. Recall that for a given wave number $n$, there are infinitely many eigenvalues. In the results that follow, $\sigma$ refers to the maximum over all eigenvalues. $\sigma_{max}$ refers to the maximum over all eigenvalues and over all wave numbers. For consistency, we often use the same parameter values throughout our results. Unless otherwise stated, $\mu_i = 2$, $\mu_o = 10$, $T_1 = T_2 = 1$, and $Q = 10$. Therefore, the inner and outer layer fluids have constant viscosity 2 and 10 respectively for all our studies here. The viscous profile of the middle layer fluid, however, is a free variable which can be taken as constant or variable in our studies below. In the rest of the paper, we will characterize the flow by the viscosity of the middle layer.

\subsection{Constant vs. Variable Viscosity}\label{sec:ConstantVsVar}
\begin{figure}[!ht]
\begin{subfigure}[b]{.5\textwidth}
  \centering
  \includegraphics[width=\textwidth, height=\textwidth]{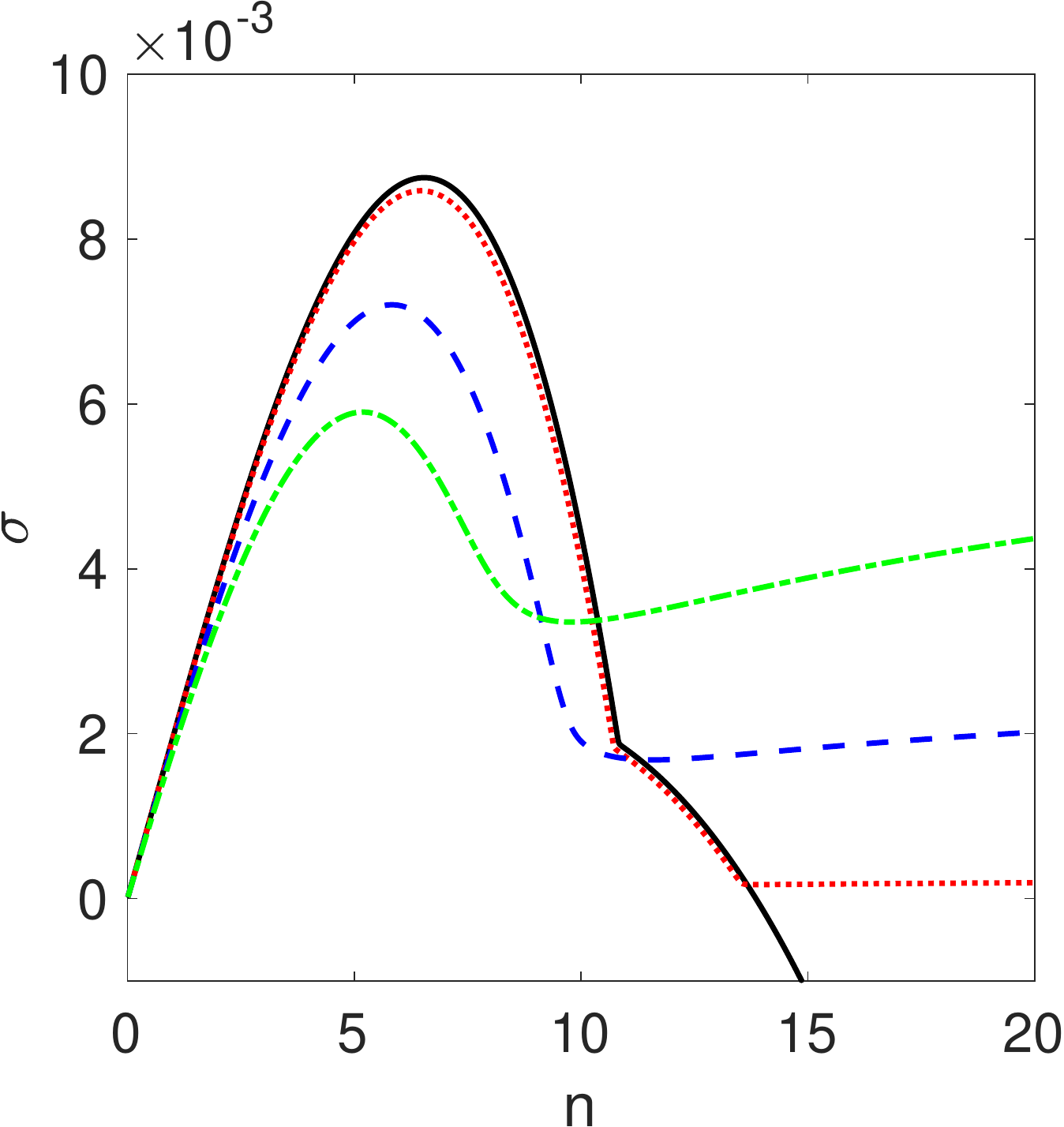}
  \caption{}
  \label{fig:ConstantVsVar_Disp}
\end{subfigure}%
\begin{subfigure}[b]{.5\textwidth}
  \centering
  \includegraphics[width=\textwidth, height=\textwidth]{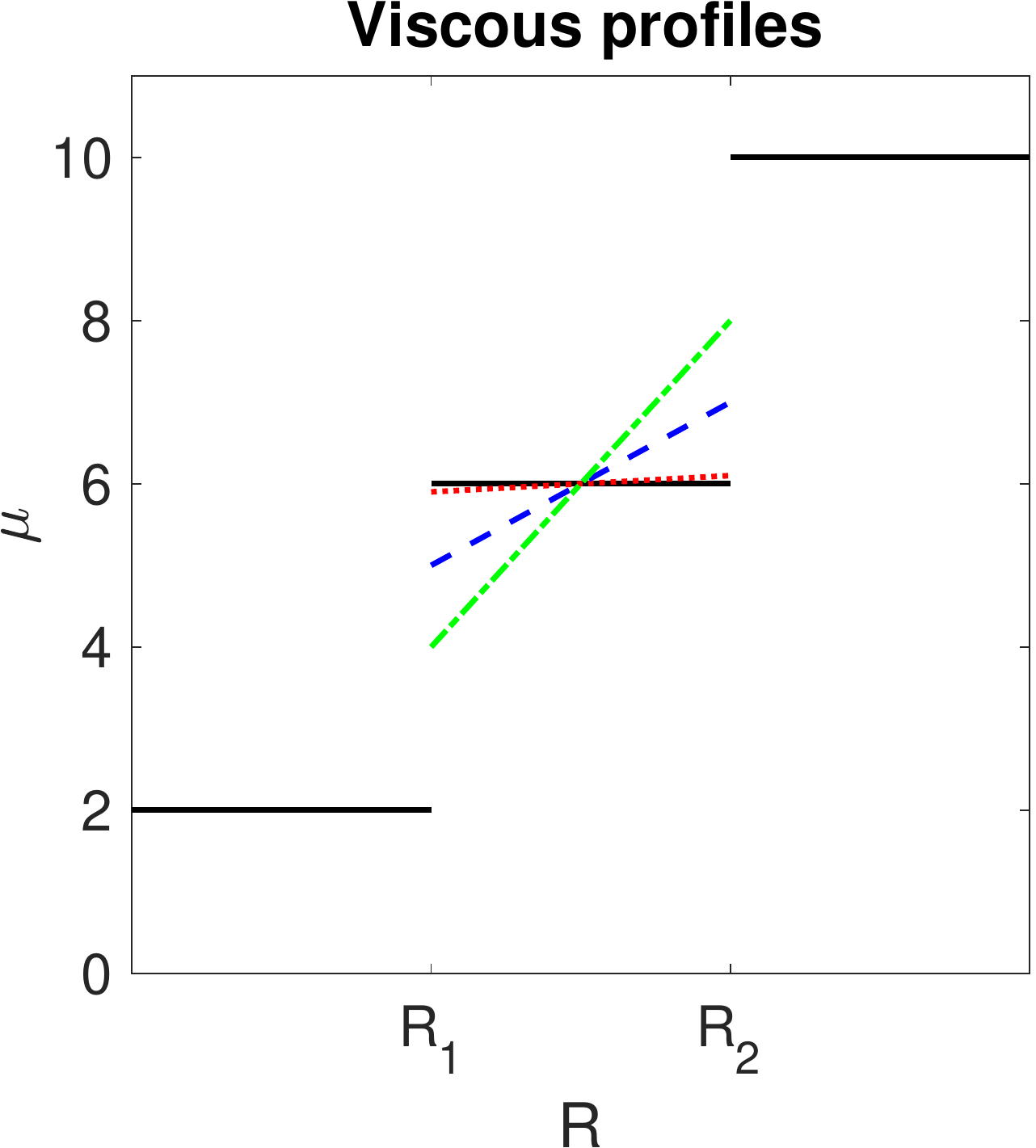}
  \caption{}
  \label{fig:ConstantVsVar_Profs}
\end{subfigure}
\caption{A comparison of dispersion relations for four linear viscous profiles. Sub figure (a) shows the dispersion relations ($\sigma$ vs. $n$) at time $\tau = 0$ and sub figure (b) depicts the associated viscous profiles. The parameter values are $Q = 10$, $\mu_i = 2$, $\mu_o = 10$, $T_1 = T_2 = 1$, $R_1(0) = 20$, and $R_2(0) = 30$.}
\label{fig:ConstantVsVar}
\end{figure}

We begin by comparing the growth rate of disturbances for a constant viscosity profile with that for a variable viscous profile. In Figure \ref{fig:ConstantVsVar_Disp}, the dispersion relations at a fixed time ($\tau = 0$) are plotted for four different viscous profiles shown in Figure \ref{fig:ConstantVsVar_Profs}. The constant viscosity case is given by the solid (black) line and the viscosity of the middle layer fluid is $\mu = 6$. The stability of three-layer constant viscosity flows has been studied extensively by \cite{gin-daripa:hs-rad, gin-daripa:constant_interfaces}. Note that there is a maximum growth rate and that short waves are stable due to interfacial tension. For comparison, three linear viscous profiles are considered. The dotted (red) line corresponds to a linear viscous profile with $\mu(R_1) = 5.9$ and $\mu(R_2) = 6.1$, the dashed (blue) line corresponds to $\mu(R_1) = 5$ and $\mu(R_2) = 7$, and the dash-dot (green) line corresponds to $\mu(R_1) = 4$ and $\mu(R_2) = 8$. There are several important features to notice. First, the dispersion relation for each of the variable viscous profiles has a local maximum for a wave number that is similar to the maximum for the case of a constant viscous profile. For profiles with smaller viscous jumps at the interfaces, the local maximum is smaller. Therefore, this local maximum can be attributed to the instability of the interfaces due to the positive viscous jump. The second thing to notice is that short waves are unstable for variable viscous profiles, even when the viscous profile is nearly constant (see the dotted line). As the gradient of the viscous profiles increase, the growth rate of short waves also increases. Therefore, the short wave behavior is dominated by the instability of the middle layer fluid itself due to the viscous gradient. The final observation which can be drawn from Figure \ref{fig:ConstantVsVar} is that the maximum growth rate can be smaller for a variable viscous profile than it is for a constant viscous profile with constant viscosity equal to the average of the values of viscosity at the two interfaces in the middle layer of the variable viscosity profile.

\subsection{Optimal Profile}
In subsection \ref{sec:ConstantVsVar}, it is shown that some particular variable viscous profiles are less unstable (i.e. have a smaller maximum growth rate) than a particular constant viscous profile. This leads to some more general questions: Are there variable viscous profiles that are less unstable than all constant viscous profiles? What is the optimal viscous profile? 

The question of the optimal viscous profile is a difficult one so we start by using some simplifying assumptions. First, recall that the viscous profile and the growth rate are both time-dependent. In this section we only consider the growth rate at time $\tau = 0$. This is reasonable because it is advantageous to control the instability at early times. Therefore, for the present discussion the term ``optimal'' refers to the viscous profile that minimizes the maximum growth rate $\sigma_{max}$ at time $\tau = 0$. The second simplification is that we first consider only viscous profiles in the middle layer that are linear at time $\tau = 0$. Note that linear profiles are uniquely determined by the values $\mu(R_1)$ and $\mu(R_2)$. Other types of viscous profiles will be considered later in this section.

  \begin{figure}[!ht]
\begin{subfigure}[b]{.5\textwidth}
  \centering
  \includegraphics[width=\textwidth, height=\textwidth]{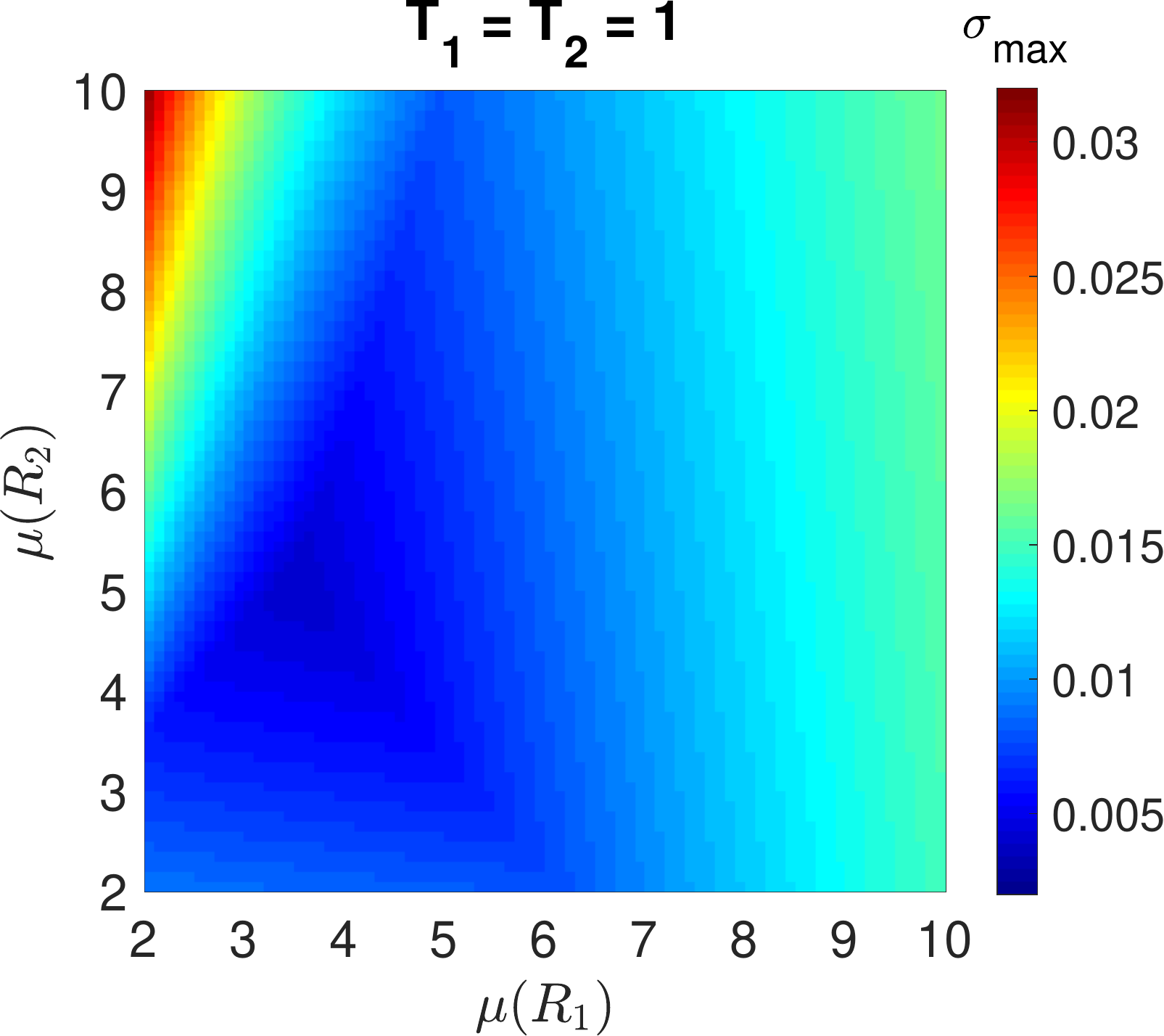}
  \caption{}
  \label{fig:Optimal_Linear_Color}
\end{subfigure}%
\begin{subfigure}[b]{.5\textwidth}
  \centering
  \includegraphics[width=\textwidth, height=\textwidth]{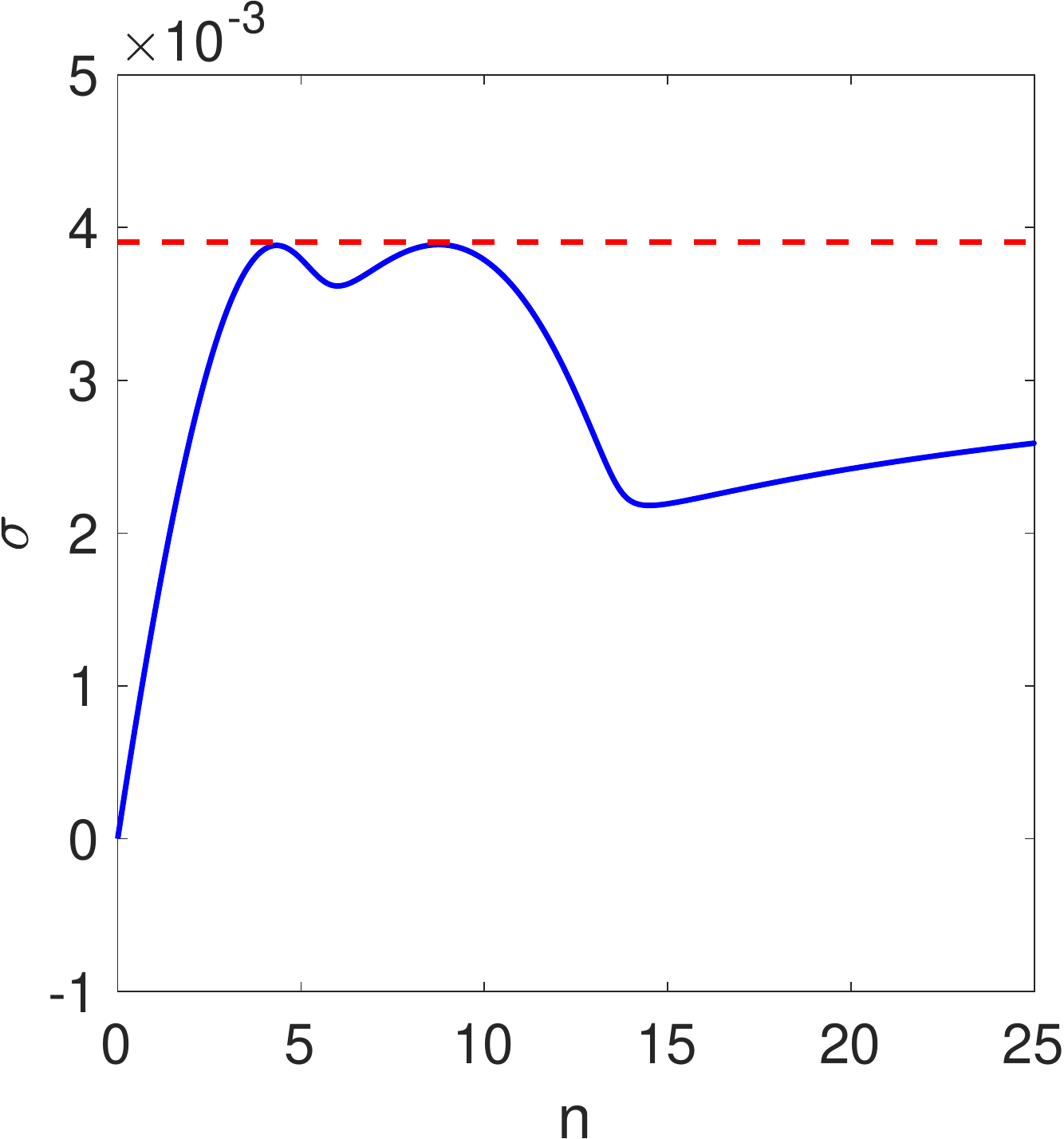}
  \caption{}
  \label{fig:Optimal_DispRel}
\end{subfigure}
\caption{Sub figure (a) gives the value of $\sigma_{max}$ for each linear viscous profile which is defined by the values $\mu(R_1)$ and $\mu(R_2)$. The optimal profile is $\mu(R_1) = 3.41$ and $\mu(R_2) = 5.09$ and its dispersion relation is given in sub figure (b). The parameter values are $Q = 10$, $\mu_i = 2$, $\mu_o = 10$, $T_1 = T_2 = 1$, $R_1(0) = 20$, and $R_2(0) = 30$.}
\label{fig:Optimal_Linear}
\end{figure}

Figure \ref{fig:Optimal_Linear_Color} shows the value of $\sigma_{max}$ for each linear viscous profile such that the viscosity of the middle layer is between $\mu_i$ and $\mu_o > \mu_i$. The optimal viscous profile in this case is $\mu(R_1) = 3.41$ and $\mu(R_2) = 5.09$. Note that all possible constant viscous profiles have been considered as a subset of the set of linear viscous profiles. Therefore, the fact that the optimal profile is not constant shows that variable viscous profiles can be used to reduce the instability of a flow. The dispersion curve for the optimal viscous profile shown in Figure \ref{fig:Optimal_Linear} approaches the value indicated by the dotted horizontal line as $n \to \infty$. The two local maxims in this plot have the same value as this limit. This is because the optimal viscous profile is the one which balances the instabilities of the interfaces with the instability of the middle layer.

Next we investigate the optimal viscous profile under several different values of interfacial tension. Plots of $\sigma_{max}$ versus the different linear profiles is given in Figure \ref{fig:ColorOptimal}.  Figure \ref{fig:ColorT25} has the smallest value of interfacial tension with $T_1 = T_2 = 0.25$. The optimal viscous profile has endpoint viscosities of $\mu(R_1) = 3.20$ and $\mu(R_2) = 5.65$.  Figure \ref{fig:ColorT1} uses $T_1 = T_2 = 1$ and is a repeat of Figure \ref{fig:Optimal_Linear_Color}. As noted above, the optimal profile is $\mu(R_1) = 3.41$ and $\mu(R_2) = 5.09$. Figure \ref{fig:ColorT4} has the largest values of interfacial tension with $T_1 = T_2 = 4$. The optimal linear viscous profile is $\mu(R_1) = 3.47$ and $\mu(R_2) = 4.53$. The trend is that larger values of interfacial tension correspond to optimal viscous profiles with a smaller viscous gradient. This is because, as mentioned previously, the optimal viscous profile is the one which balances the instabilities of the interfaces with the instability of the middle layer. A larger value of interfacial tension decreases the instability of the interfaces so the gradient of the middle layer must also decrease in order to match the interfacial instability. 

 \begin{figure}[!ht]
\begin{subfigure}[b]{.32\textwidth}
  \centering
  \includegraphics[width=\textwidth, height=\textwidth]{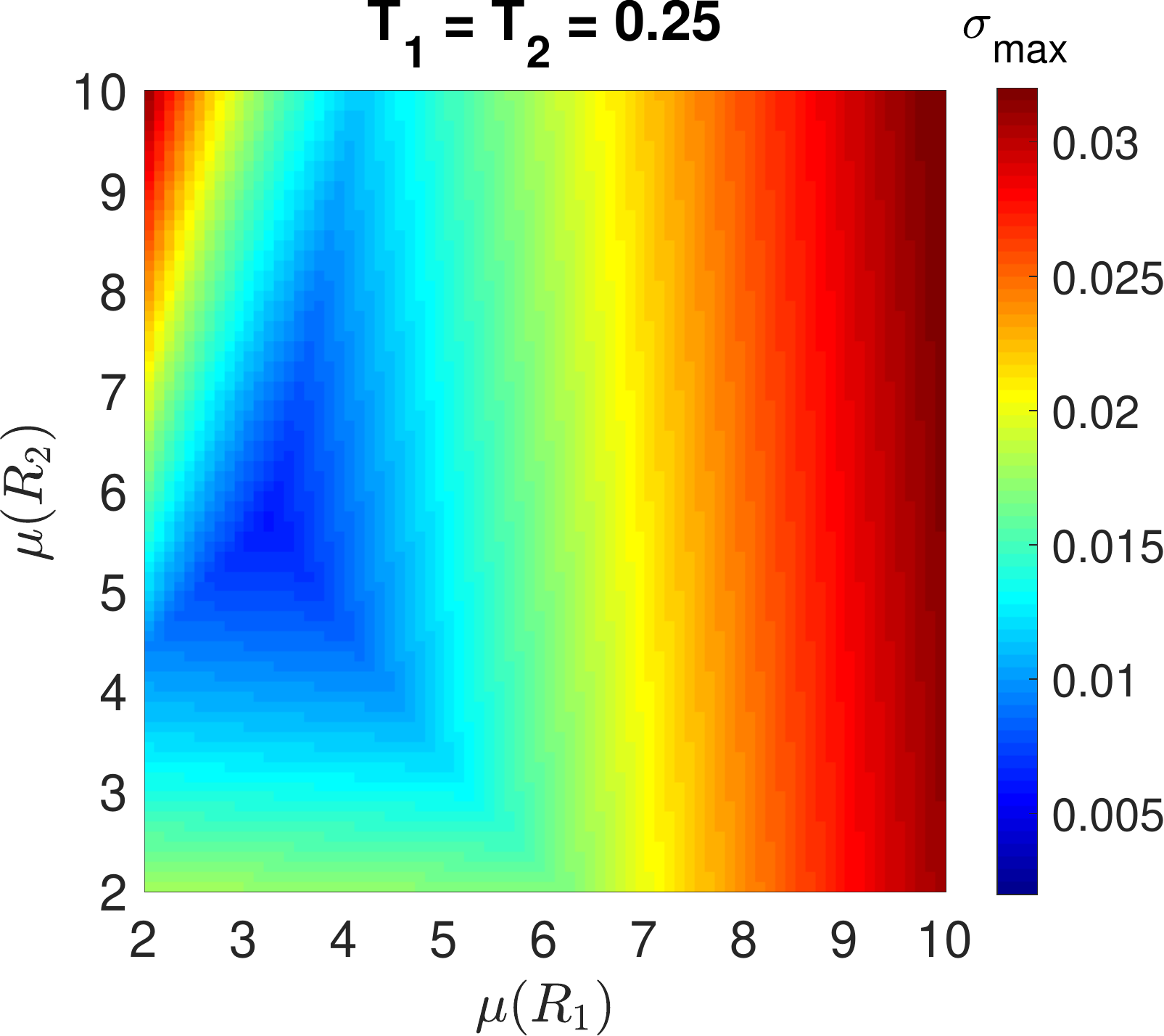}
  \caption{}
  \label{fig:ColorT25}
\end{subfigure}%
\begin{subfigure}[b]{.32\textwidth}
  \centering
  \includegraphics[width=\textwidth, height=\textwidth]{ColorT1.pdf}
  \caption{}
  \label{fig:ColorT1}
\end{subfigure}
\begin{subfigure}[b]{.32\textwidth}
  \centering
  \includegraphics[width=\textwidth, height=\textwidth]{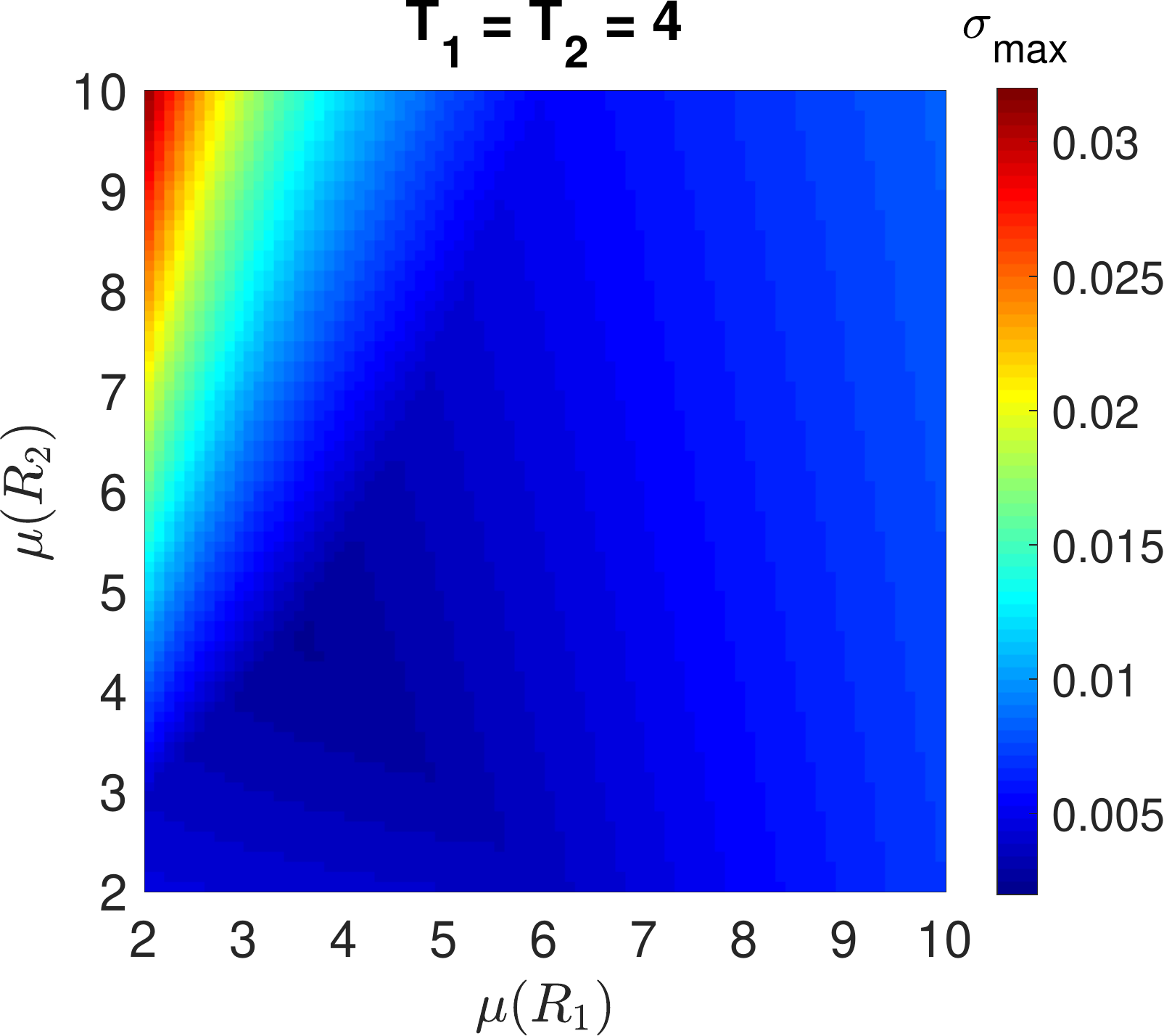}
  \caption{}
  \label{fig:ColorT4}
\end{subfigure}
\caption{Plots of $\sigma_{max}$ for all linear viscous profiles for three different values of interfacial tension: (a) $T_1 = T_2 = 0.25$, (b) $T_1 = T_2 = 1$, and (c) $T_1 = T_2 = 4$. Other parameter values are $Q = 10$, $\mu_i = 2$, $\mu_o = 10$, $R_1(0) = 20$, and $R_2(0) = 30$.}
\label{fig:ColorOptimal}
\end{figure}

In the case of chemical EOR by polymer flooding, the viscous profile of the middle layer fluid is determined by the concentration of polymer. The use of large quantities of polymer can be expensive so it is useful to explore which viscous profile minimizes the instability (i.e. $\sigma_{\rm max}$) given a fixed amount of total polymer. Assuming that there is a linear relationship between the concentration of polymer and the viscosity, this can be viewed as minimizing the instability for a fixed value of average viscosity. The results of this type of optimization are given in  Figure \ref{fig:Optimal_Fixed_Concentration}. The value of $\sigma_{max}$ for all linear viscous profiles is plotted in Figure \ref{fig:Optimal_Fixed_Concentration_Color} using the same parameter values as  Figure \ref{fig:Optimal_Linear_Color}. For each possible average viscosity between $\mu_i$ and $\mu_o$, the profile which minimizes $\sigma_{max}$ was found and is marked by an `$x$' in Figure \ref{fig:Optimal_Fixed_Concentration_Color}. The viscous profile can be identified by its slope $a = (\mu(R_2) - \mu(R_1))/(R_2-R_1)$. The slopes of the optimal profiles are plotted versus the average viscosity of the middle layer in Figure \ref{fig:Optimal_Fixed_Concentration_Slope}.

  \begin{figure}[!ht]
\begin{subfigure}[b]{.5\textwidth}
  \centering
  \includegraphics[width=\textwidth, height=\textwidth]{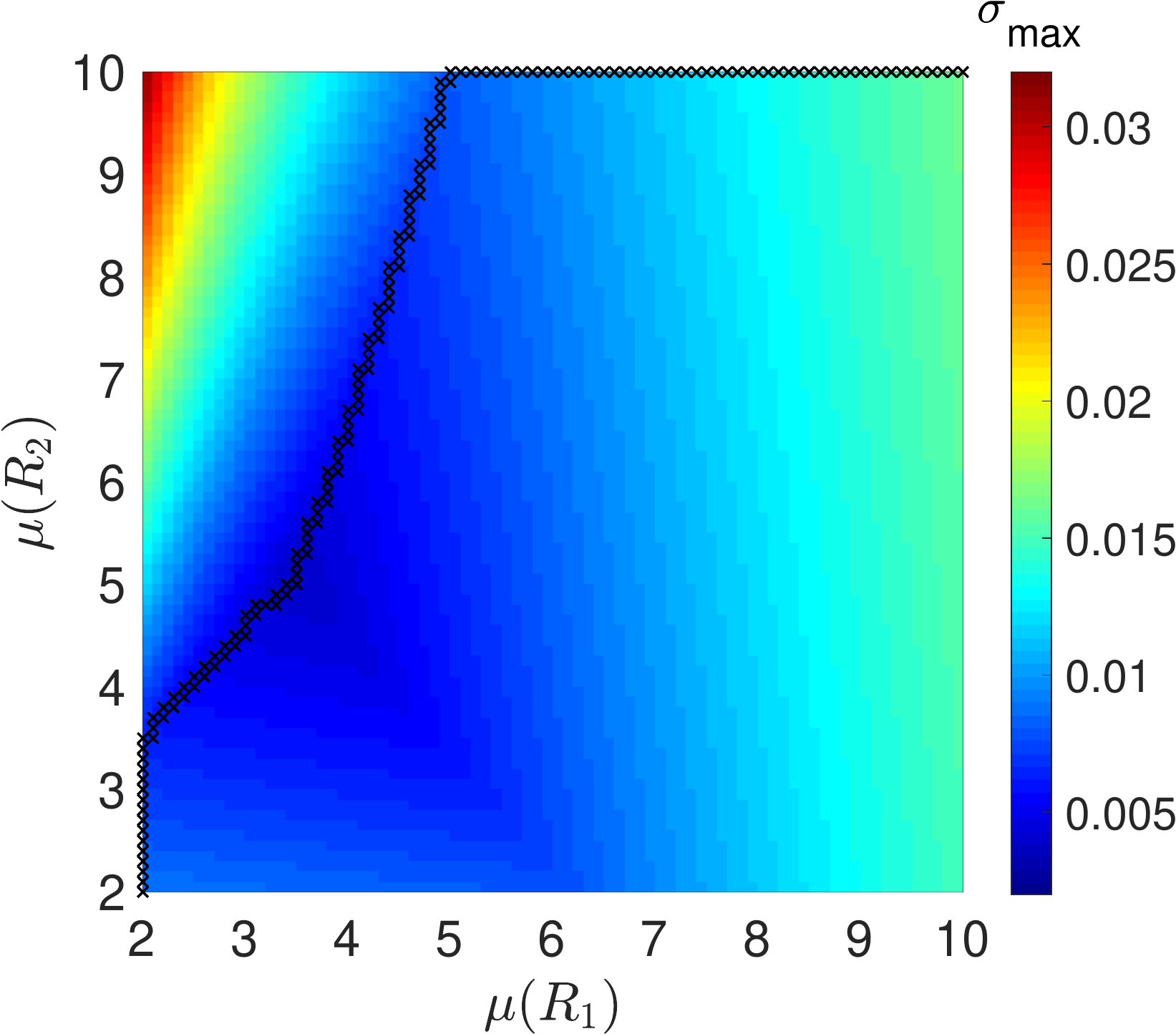}
  \caption{}
  \label{fig:Optimal_Fixed_Concentration_Color}
\end{subfigure}%
\begin{subfigure}[b]{.5\textwidth}
  \centering
  \includegraphics[width=\textwidth, height=\textwidth]{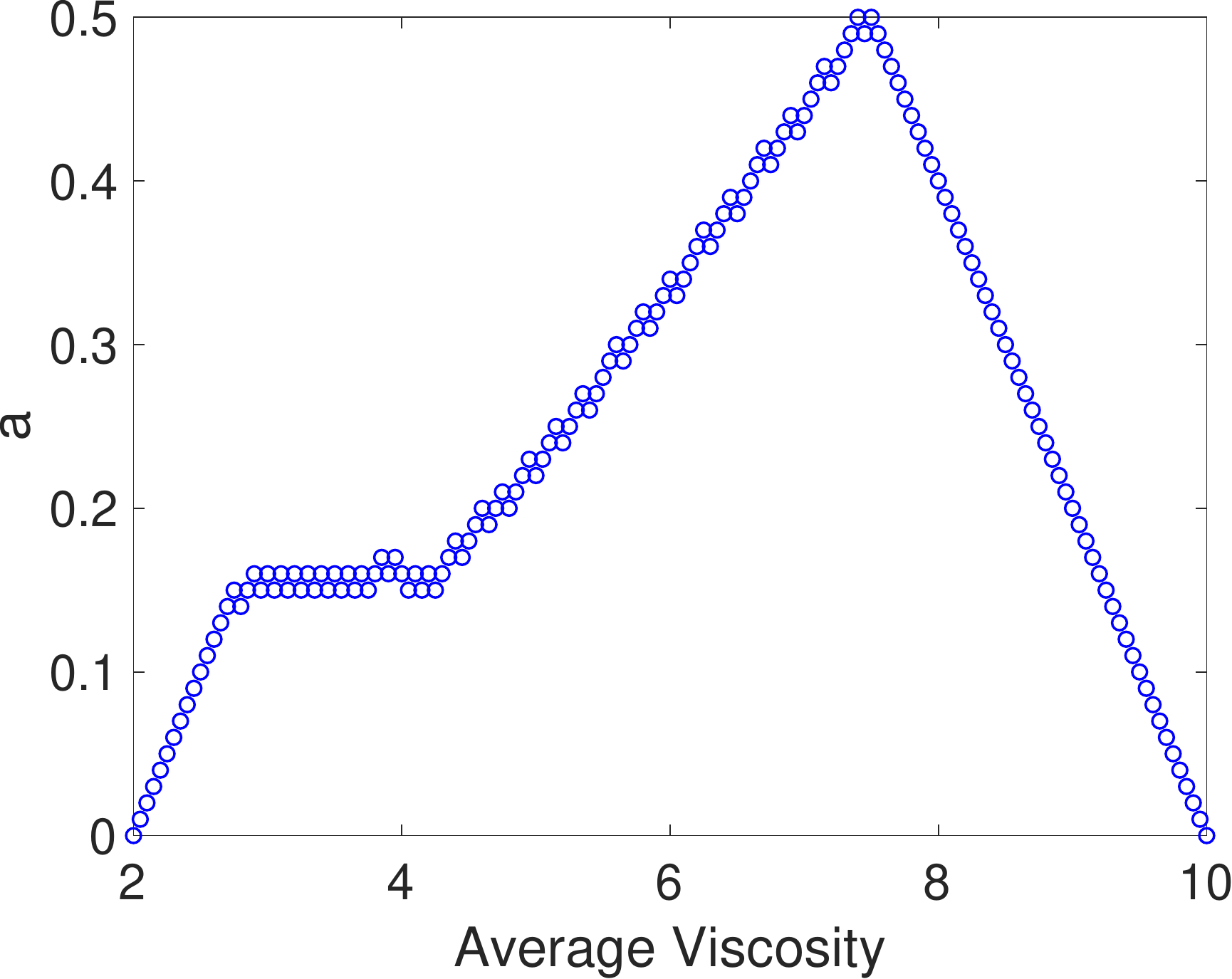}
  \caption{}
  \label{fig:Optimal_Fixed_Concentration_Slope}
\end{subfigure}
\caption{Plots of the optimal linear viscous profile for a fixed value of average viscosity of the middle layer fluid. Sub figure (a) shows the value of $\sigma_{max}$ versus the different linear viscous profiles with $x$'s to denote the optimal profiles. Sub figure (b) is a plot of the slope of the optimal profile versus the average viscosity. The parameter values are $Q = 10$, $\mu_i = 2$, $\mu_o = 10$, $T_1 = T_2 = 1$, $R_1(0) = 20$, and $R_2(0) = 30$.}
\label{fig:Optimal_Fixed_Concentration}
\end{figure}

When the average viscosity of the middle layer fluid is $\mu_i = 2$, the optimal viscous profile is constant at $\mu \equiv 2$ (note that this is the only profile considered since we are taking $\mu_i \leq \mu(r) \leq \mu_o$ for $R_1 \leq r \leq R_2$). Therefore there is no viscous jump at the inner interface and no instability in the layer itself. All of the instability occurs due to the jump in viscosity at the outer interface. As the average viscosity increases from there, the jump at the inner interface of the optimal viscous profile remains zero while the viscous gradient increases in order to decrease the viscous jump at the outer interface. Eventually, as the average viscosity nears $\mu = 3$, a point is reached in which the gradient stops increasing as illustrated by the flat portion of the graph in Figure \ref{fig:Optimal_Fixed_Concentration_Slope}. During this time, the viscous jump at the inner interface increases while the viscous jump at the outer interface decreases. The point at which the slope begins to increase again corresponds to the optimal viscous profile over all values of average viscosity. After this point, the addition of polymer to increase viscosity would be detrimental to the stability of the system.

\begin{figure}[!ht]
\begin{subfigure}[b]{.5\textwidth}
  \centering
   \includegraphics[width=\textwidth, height=\textwidth]{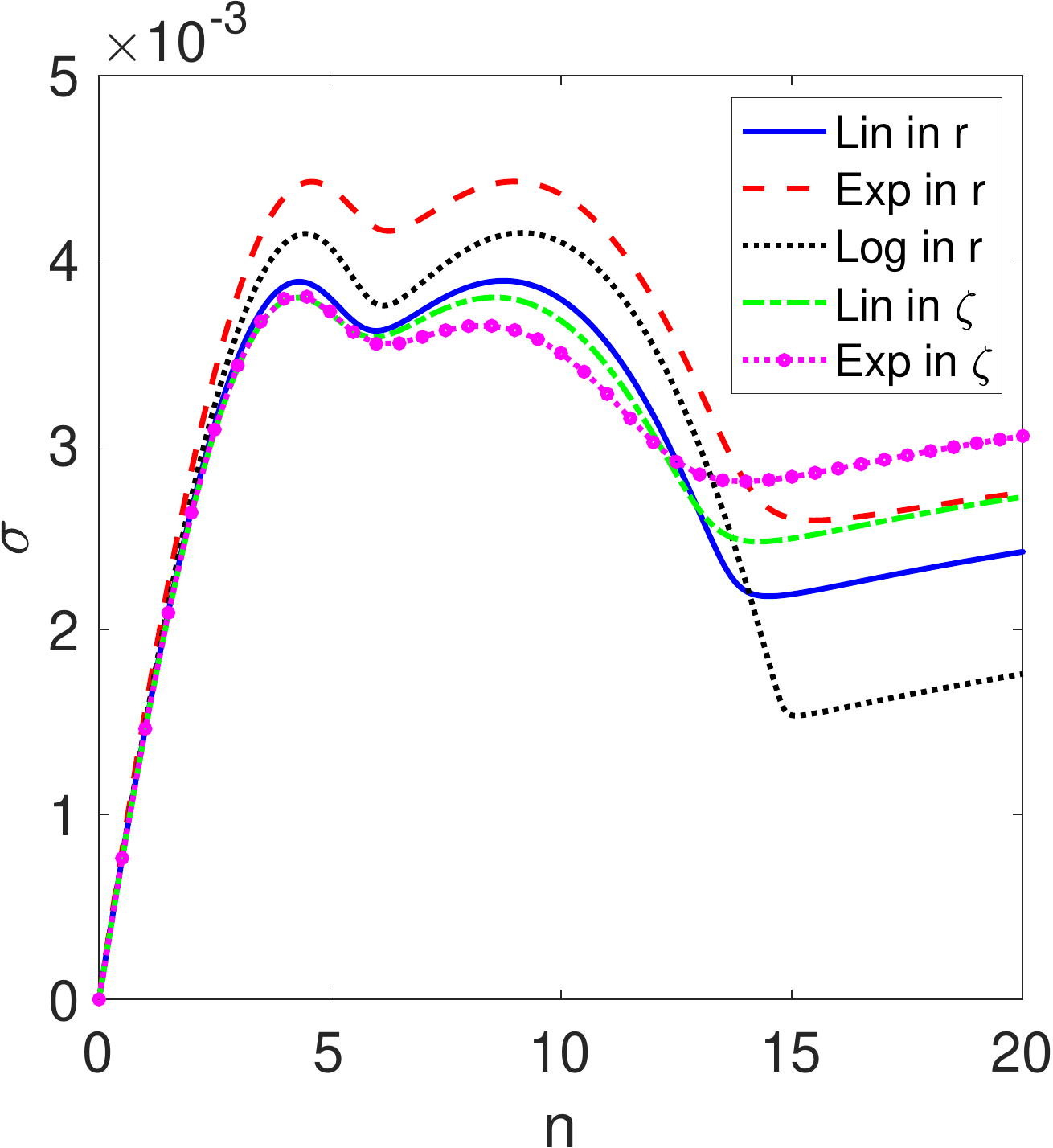}
  \caption{}
 \label{fig:Compare_Opt_Profiles}
\end{subfigure}%
\begin{subfigure}[b]{.5\textwidth}
  \centering
   \includegraphics[width=\textwidth, height=\textwidth]{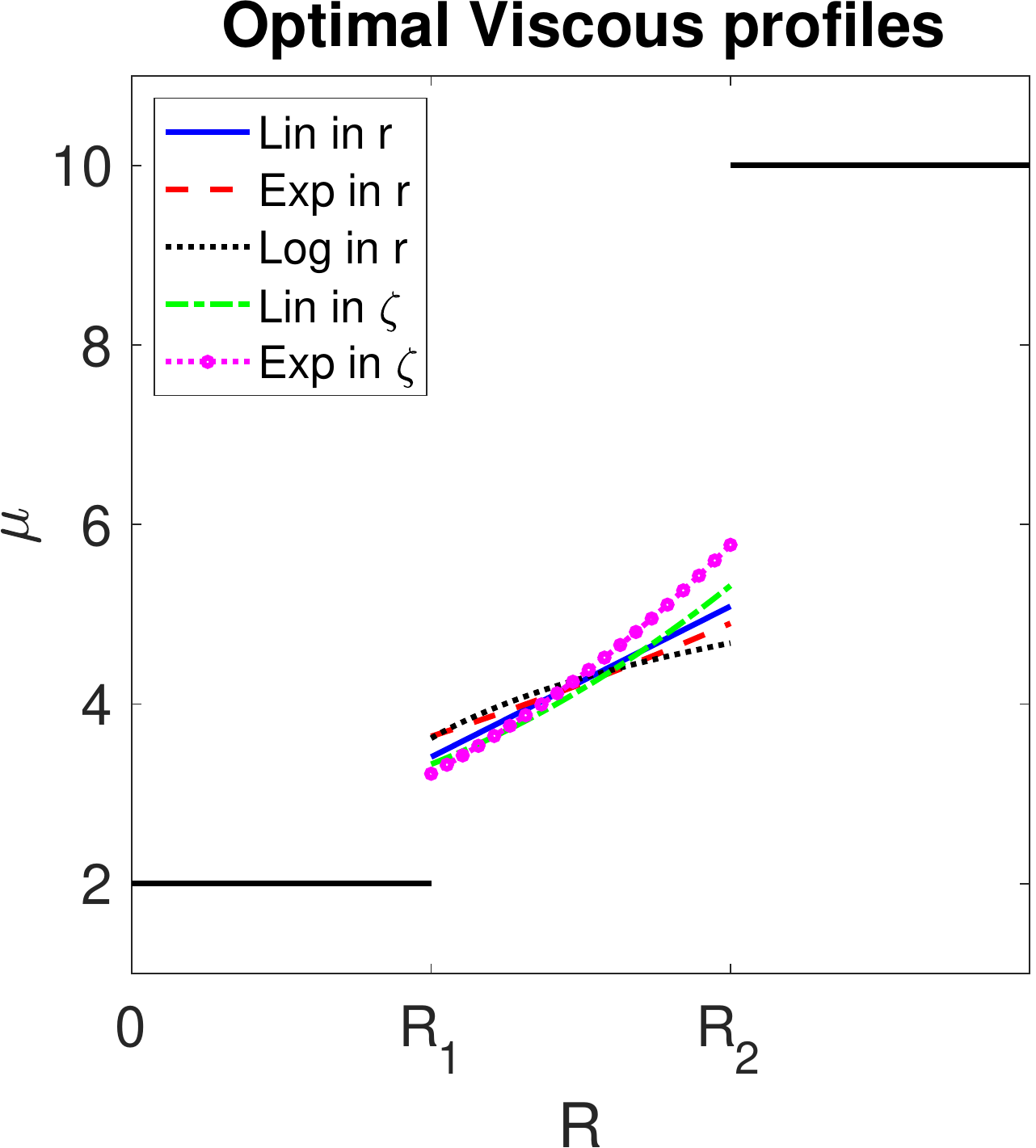}
  \caption{}
  \label{fig:Compare_Opt_Profiles2}
\end{subfigure}
 \caption{(a) Plots of the dispersion relations for the optimal viscous profiles which are (i) linear with respect to $r$,  (ii) exponential with respect to $r$,  (iii) logarithmic with respect to $r$,  (iv) linear with respect to $\zeta$, and (v) exponential with respect to $\zeta$. (b) Plots of the corresponding viscous profiles.}
 \label{fig:Compare_Opt}
\end{figure}

 To this point, only linear viscous profiles have been considered, but many other viscous profiles can also be used. The optimization procedure used for linear viscous profiles in  Figure \ref{fig:Optimal_Linear_Color} has been repeated for several other types of viscous profiles in  Figure \ref{fig:Compare_Opt}. In addition to a viscous profile which is linear at time $\tau = 0$, viscous profiles which are initially exponential and logarithmic are also considered. Also, recall that the viscous profile changes with time in the $r$-coordinates. Therefore, it may be useful to consider viscous profiles in the $\zeta$-coordinate system because they will remain fixed in time. Therefore, we also consider viscous profiles which are linear and exponential with respect to $\zeta$. A profile which is linear with respect to $\zeta$ is quadratic with respect to $r$ and a profile which is exponential in $\zeta$ is proportional to $e^{r^2}$. Figure \ref{fig:Compare_Opt_Profiles} shows the dispersion relations of the optimal viscous profiles of each type. The profile which is exponential in $r$ is the most unstable while the profile which is exponential in $\zeta$ is the least unstable. The corresponding optimal viscous profiles are plotted in Figure \ref{fig:Compare_Opt_Profiles2}. Notice that the profile which is least unstable, the one which is exponential in $\zeta$, has both the smallest value of $\mu(R_1)$ and the largest value of $\mu(R_2)$. Therefore, it has the smallest viscous jumps at the interfaces.

\subsection{Time dependence}
In the previous sections, we considered the growth rate only at time $\tau=0$. However, it is also important to understand how the growth rate changes with time. As time increases and the interfaces move outward, there are several physical factors at play. The curvature of the interfaces decreases which works to destabilize the flow while the velocity of the interfaces decreases which works to stabilize the flow. \cite{gin-daripa:hs-rad} studied constant viscosity flows and found that the above two competing effects lead to non-monotonic behavior of the maximum growth rate. However, the growth rate of constant viscosity flows behaves differently in the $\zeta$-coordinate system. We illustrate this first with some analytical results.

For two-layer constant viscosity flows, the growth rate in the $\zeta$-coordinates (see equation \eqref{VariableViscosity:2layerGR}) is
\begin{equation*}
 \sigma = \frac{Qn}{2 \pi R^2} \frac{\mu_o-\mu_i}{\mu_o+\mu_i} - \frac{T}{\mu_o+\mu_i} \frac{n^3-n}{R^3},
\end{equation*}
where $R(\tau)$ is the radius of the interface. By taking a derivative with respect to $n$ and setting equal to zero, the most dangerous wave number is \\$n_{max} = \sqrt{QR/(6 \pi T) (\mu_o-\mu_i) + 1/3}$. Plugging this into \eqref{VariableViscosity:2layerGR} gives 
\begin{equation*}
\begin{split}
 \sigma_{max} = &\frac{Q}{2 \pi R^2} \sqrt{\frac{QR}{6 \pi T}(\mu_o - \mu_i) + \frac{1}{3}} \left(\frac{\mu_o-\mu_i}{\mu_o+\mu_i}\right) \\
  - &\frac{T}{\mu_o+\mu_i} \sqrt{\frac{QR}{6 \pi T}(\mu_o - \mu_i) + \frac{1}{3}} \left(\frac{QR}{6 \pi T}(\mu_o - \mu_i) - \frac{2}{3}\right) \frac{1}{R^3}.
 \end{split}
\end{equation*}
Taking the derivative with respect to $R$ gives
\begin{equation*}
 \frac{\partial \sigma_{max}}{\partial R} = - \frac{T \left(\frac{Q}{\pi T} (\mu_o-\mu_i)R + 2 \right) \left(\frac{Q}{\pi T} (\mu_o-\mu_i)R+4 \right)}{12 R^4 (\mu_i+\mu_o) n_{max}}.
\end{equation*}
If $\mu_o > \mu_i$ then this expression is negative for all $R$. Therefore, $\sigma_{max}$ is a strictly decreasing function of $R$ and hence time $\tau$ in the $\zeta$-coordinates. 

\begin{figure}[!ht]
 \centering
 \includegraphics[width=.5\textwidth, height=.5\textwidth]{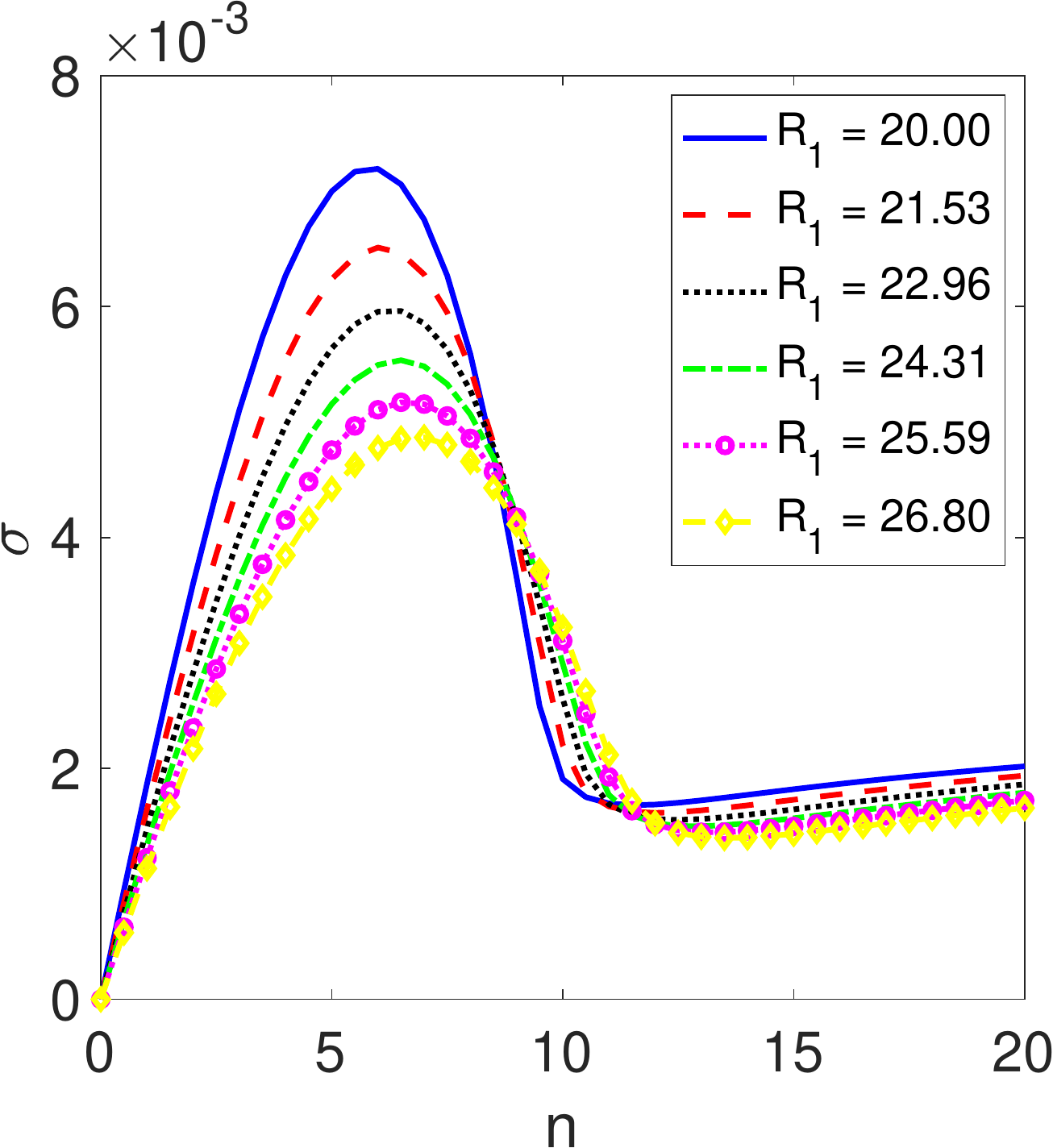}
 \caption{Several plots of the dispersion relation $\sigma$ vs $n$ at different times, as represented by the location of the inner interface. Parameter values are $Q = 10$, $\mu_i = 2$, $\mu_o = 10$, $\mu(R_1) = 5$, $\mu(R_2) = 7$, $T_1 = T_2 = 1$, $R_1(0) = 20$, and $R_2(0) = 30$.}
 \label{fig:DispRel_6times}
\end{figure}

For three-layer variable viscosity flows, there is an additional factor which affects the stability. The interfaces get closer together which makes the variable viscous profile steeper and works to destabilize the flow. Despite this fact, the numerical results that follow show that $\sigma_{max}$ is a decreasing function of time. This can be illustrated by the upper bound given by \eqref{VariableViscosity:AbsUpperBound}. The first two terms are strictly decreasing functions of $R_1$ and $R_2$ (and therefore of $\tau$) while the third term is independent of time. Therefore, the upper bound is a decreasing function of time. In Figure \ref{fig:DispRel_6times}, the dispersion relation is plotted at several different times for a typical variable viscosity flow. Initially, the inner interface is at $R_1 = 20$. As time increases, it moves outward. Note that the maximum value of $\sigma$ decreases with time. However, the difference is more pronounced near the maximum value, which is mostly affected by the stability of the interfaces, than for short waves which are mostly affected by the layer instability.

\begin{figure}[!ht]
\begin{subfigure}[b]{.49\textwidth}
  \centering
   \includegraphics[width=\textwidth, height=\textwidth]{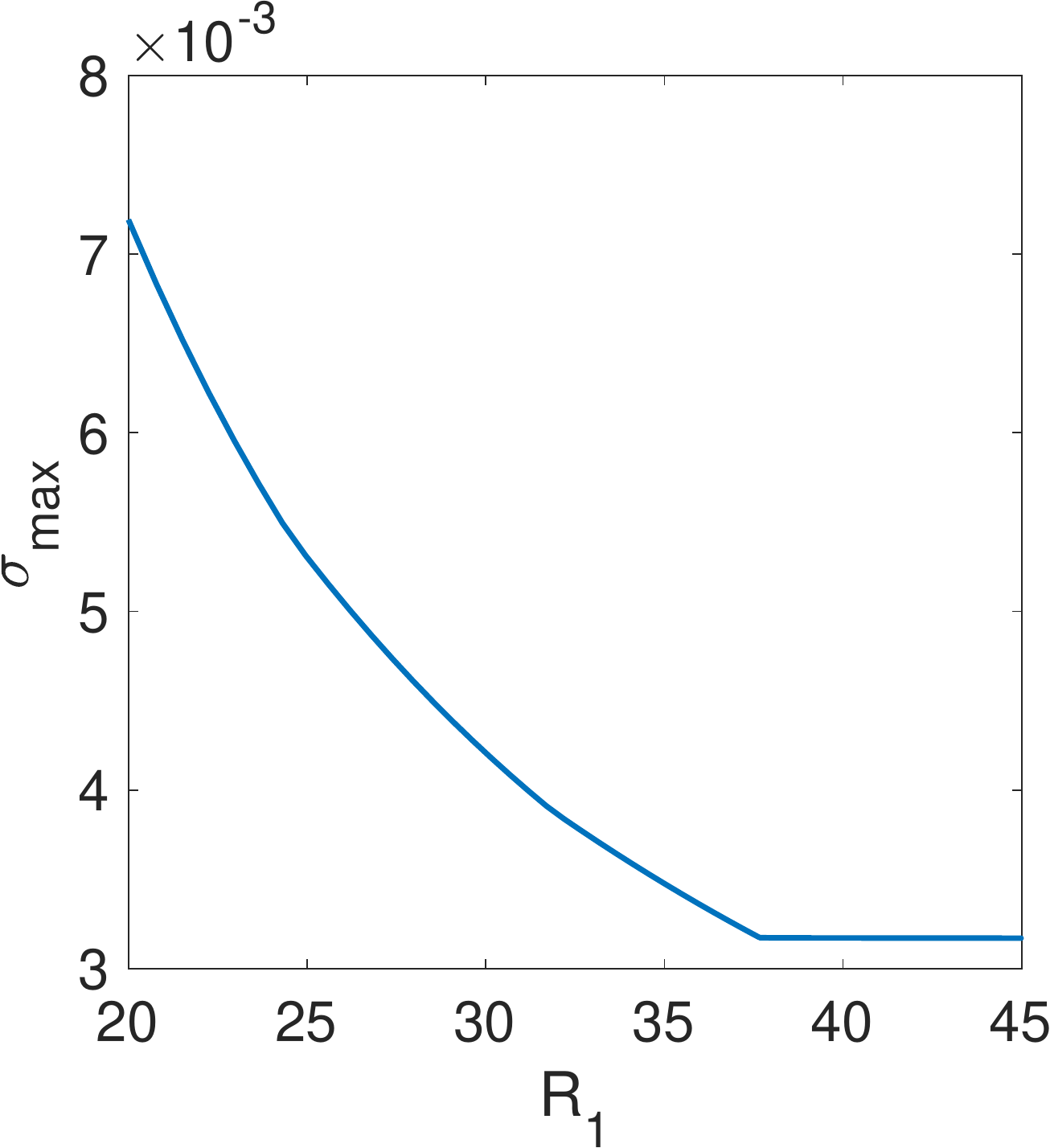}
  \caption{}
 \label{fig:SigmaMaxVsR1}
\end{subfigure}
\begin{subfigure}[b]{.49\textwidth}
  \centering
   \includegraphics[width=\textwidth, height=\textwidth]{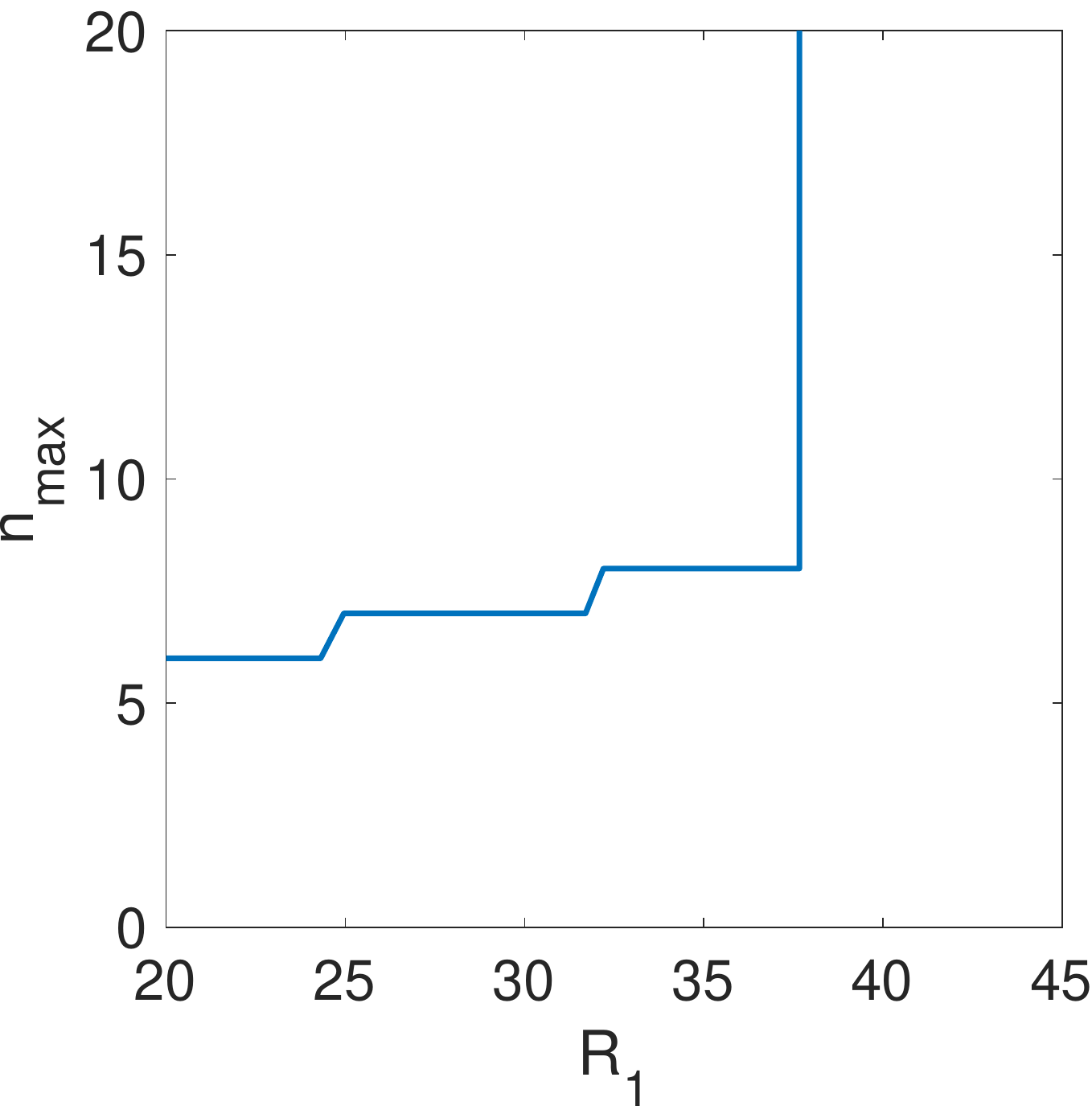}
  \caption{}
  \label{fig:nMaxVsR1}
\end{subfigure}
 \caption{Plots of (a) the maximum growth rate $\sigma_{max}$ versus the radius of the inner interface $R_1$ and (b) the most dangerous wave number $n_{max}$ versus the radius of the inner interface $R_1$. The parameter values are $Q = 10$, $\mu_i = 2$, $\mu_o = 10$, $\mu(R_1) = 5$, $\mu(R_2) = 7$, $T_1 = T_2 = 1$, $R_1(0) = 20$, and $R_2(0) = 30$.}
 \label{fig:MaxVsR1}
\end{figure}

In order to shed more light on this time-dependent behavior, we investigate how $\sigma_{max}$ and the most dangerous wave number $n_{max}$ evolve in time. Figure \ref{fig:SigmaMaxVsR1} is a plot of $\sigma_{max}$ versus $R_1$, and Figure \ref{fig:nMaxVsR1}  is a plot of $n_{max}$ versus $R_1$. Notice first that $\sigma_{max}$ is a decreasing function of time and $n_{max}$ is an increasing function of time. The fact that $n_{max}$ increases with time is a well-known fact for constant viscosity flows \citep{Cardoso/Woods:1995}. Also observe that there is a critical value $R_1^*$ such that for $R_1 > R_1^*$, $\sigma_{max}$ is constant and $n_{max}$ is infinite. This is the point at which the layer instability comes to dominate the flow. For $R_1 < R_1^*$, the instability of the interfaces dominates and the behavior or $\sigma_{max}$ and $n_{max}$ is similar to what happens for constant viscosity flow. For $R_1 > R_1^*$, the layer is more unstable than the interfaces, and therefore the short wave instability dominates and $\sigma_{max} = \displaystyle{\lim_{n \to \infty} \sigma(n)}$ . This behavior mimics what we see from the upper bound \eqref{VariableViscosity:AbsUpperBound} in which the two terms related to the interfaces are decreasing functions while the term related to the layer instability is constant.

\subsection{Validation of QSSA}
In Section \ref{sec:prelim}, we invoke the QSSA which assumes that the basic solution changes slowly in comparison to the growth of the disturbances. We now examine the validity of that assumption. Consider a two-layer constant viscosity flow in which the growth rate of a disturbance with wave number $n$ is given by equation \eqref{VariableViscosity:2layerGR}. The maximum growth rate over all wave numbers can be written as
\begin{equation*}
 \sigma_{max} = \frac{2T}{(\mu_o+\mu_i)R^3} \left( \frac{QR}{6 \pi T} (\mu_o-\mu_i) + \frac{1}{3} \right)^{\frac{3}{2}}.
\end{equation*}
The expression for the position of the interface is $R(\tau) = \sqrt{Q\tau/\pi + R(0)^2}$. In particular, $R \propto \sqrt{\tau}$. Therefore, $\sigma_{max} \propto \tau^{-3/4}$ for $\tau \gg 1$. In comparison, the interfacial position of the basic solution changes at a rate
\begin{equation*}
 \frac{1}{R} \frac{dR}{d\tau} = \frac{Q}{2 \pi R^2} \propto \tau^{-1}.
\end{equation*}
Therefore, for large $\tau$ the disturbances will grow faster than the basic solution. For three-layer variable viscosity flow, the QSSA has a more solid foundation. Recall from the previous subsection that the interfacial instability dominates at early times, but the instability of the middle layer dominates at later times. The upper bound \eqref{VariableViscosity:AbsUpperBound} demonstrates that the layer instability is bounded by a constant term that depends only on $\mu'(\zeta)$. This is further validated by the region of Figure \ref{fig:SigmaMaxVsR1} in which $\sigma_{max}$ is constant. Therefore, for variable viscosity flows the interfaces will be moving very slowly at later times while the growth of disturbances remains constant. 

Figure \ref{fig:QSSA_validation} is a numerical comparison of the growth rate of disturbances $\sigma_{max}$ and the rates of change of each individual interface of the base flow. The parameters used are the same as in Figure \ref{fig:MaxVsR1} and therefore the solid curve is the same as the curve in Figure \ref{fig:SigmaMaxVsR1}. Notice that the growth rate of the disturbances is always greater than the rate of change of the interfaces, but that this is especially true for later times. The late time behavior will be true even if the interfaces are stabilized by very large interfacial tension.
\begin{figure}[!ht]
 \centering
 \includegraphics[width=.5\textwidth]{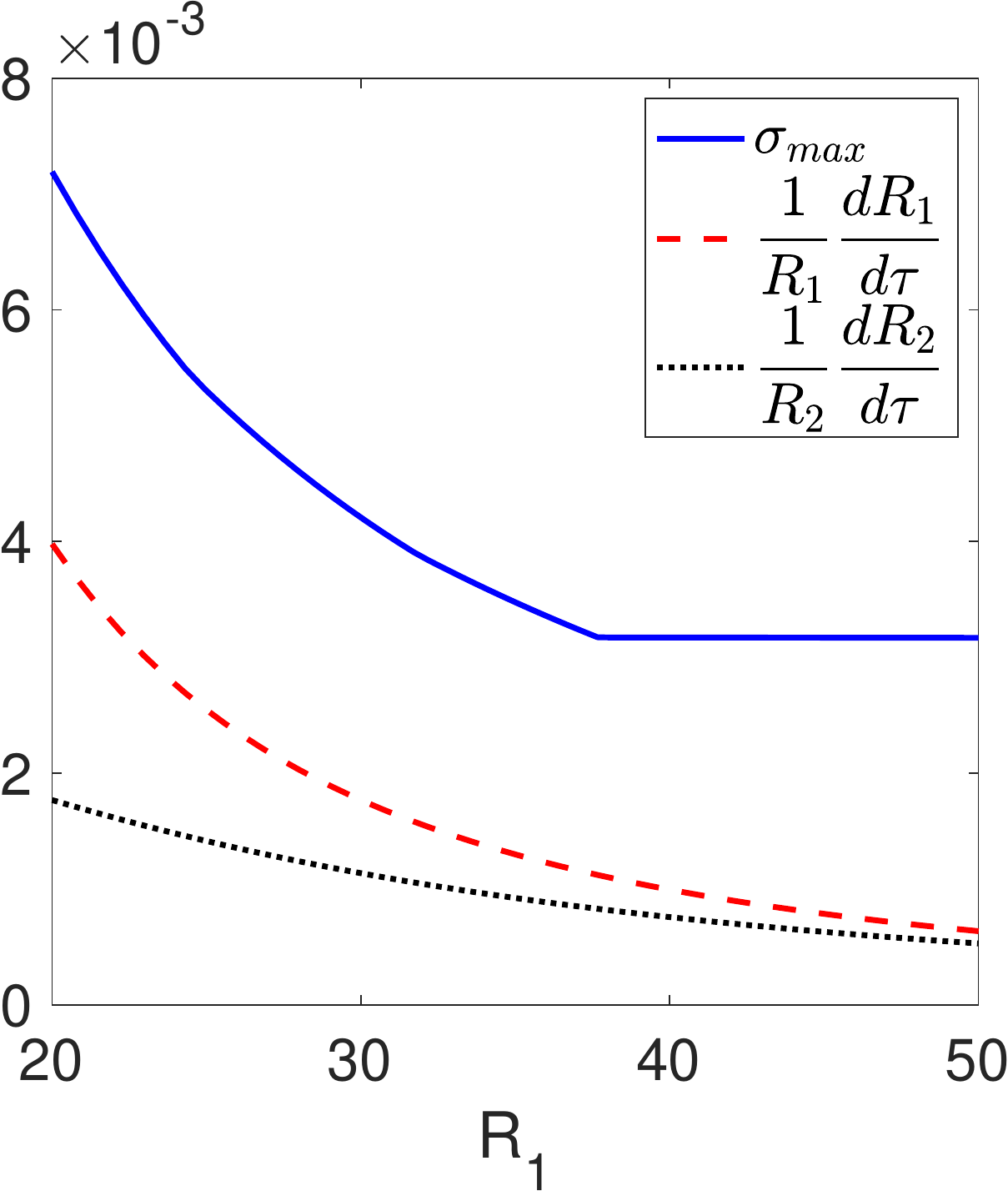}
 \caption{Plots of the maximum growth rate $\sigma_{max}$ and the rates of change of the interfaces of the basic flow. All parameter values are the same as Figure \ref{fig:MaxVsR1}: $Q = 10$, $\mu_i = 2$, $\mu_o = 10$, $\mu(R_1) = 5$, $\mu(R_2) = 7$, $T_1 = T_2 = 1$, $R_1(0) = 20$, and $R_2(0) = 30$.}
 \label{fig:QSSA_validation}
\end{figure}

\subsection{Variable Injection Rate}
Recent studies by \cite{Beeson-Jones/Woods:2015} and \cite{gin-daripa:constant_interfaces} explore the idea of using a variable injection rate $Q(t)$ to stabilize multi-layer constant viscosity flows. In these works, the maximum injection rate which results in a stable flow is investigated. Unfortunately, there is no injection rate which stabilizes a variable viscosity flow because short waves are always unstable. However, as an analogy, we can find the maximum injection rate that keeps the growth rate under a certain threshold. Figure \ref{fig:varQvstime} shows the maximum injection rate such that the maximum growth rate is below 0.001 for a certain constant viscosity flow and a certain variable viscosity flow. The constant viscosity flow has a viscosity of $\mu = 6$ in the middle layer while the variable viscosity flow has a linear viscous profile with $\mu(R_1) = 5$ and $\mu(R_2) = 7$. The variable viscosity flow allows for the fluid to be injected more quickly while maintaining the same level of instability.

\begin{figure}[!ht]
 \centering
 \includegraphics[width=.5\textwidth, height=.5\textwidth]{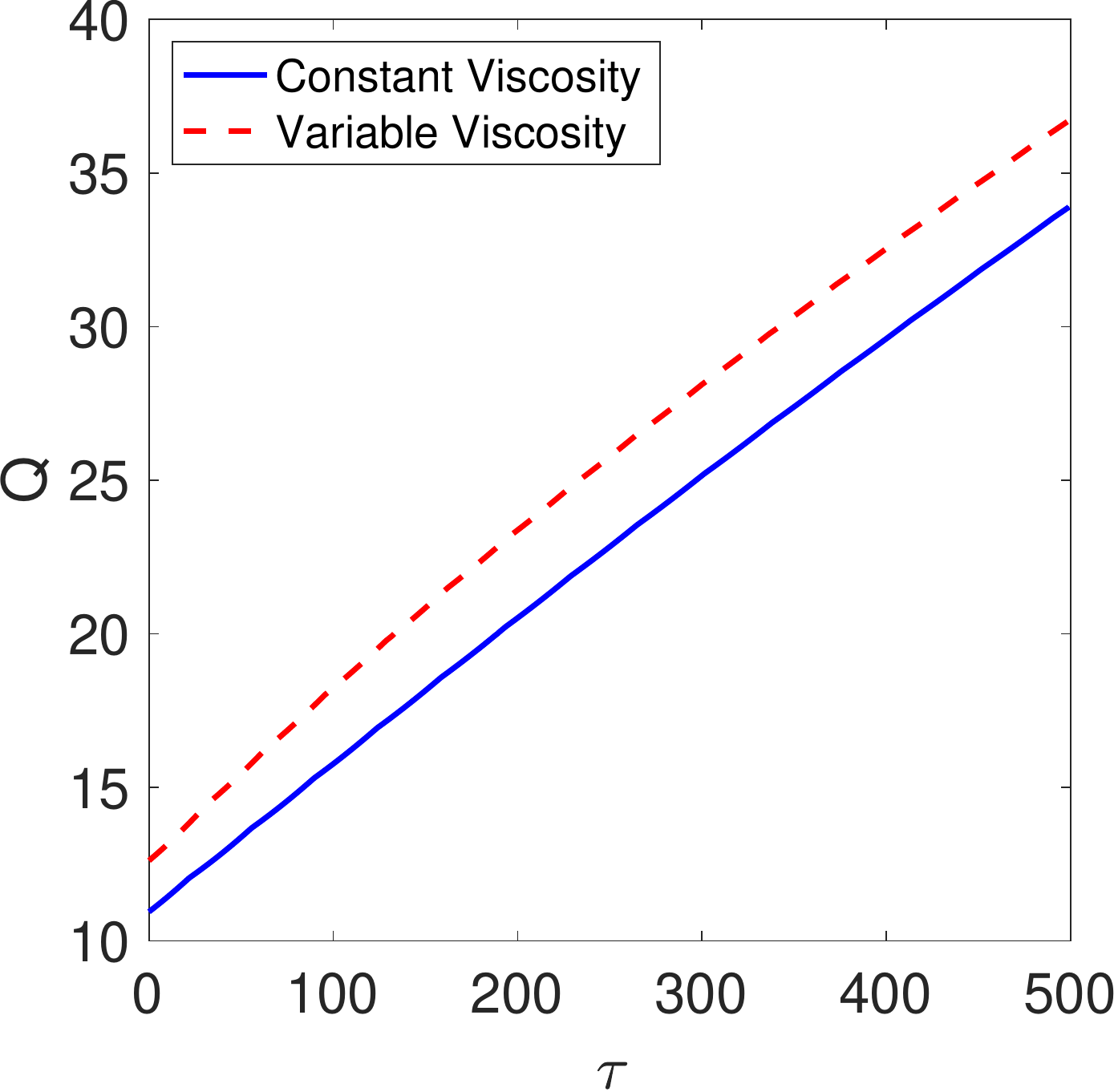}
 \caption{Plots of the maximum injection rate that results in a value of $\sigma_{max} \leq 0.001$ for a constant viscosity flow with $\mu = 6$ and a variable viscosity flow with $\mu(R_1) = 5$ and $\mu(R_2) = 7$. Other parameter values are $\mu_i = 2$, $\mu_o = 10$,  $T_1 = T_2 = 1$, $R_1(0) = 20$, and $R_2(0) = 30$.}
 \label{fig:varQvstime}
\end{figure}

\section{Conclusions}\label{sec:conclusions}
The stability of three-layer radial porous media flows with variable viscosity is an important issue in many applications. This work is the first to address this topic. First, the linear stability problem is formulated resulting in an eigenvalue problem with time-dependent coefficients and eigenvalue-dependent boundary conditions. This derivation requires an appropriate change of variables that fixes the positions of the interfaces and the viscous profile of the middle layer fluid. Several important analytical aspects of the eigenvalue problem are studied. First, upper bounds on the spectrum are derived using a variational approach. Then it is shown that for a certain bandwidth of wave numbers, there is a countably infinite set of positive eigenvalues, and the corresponding eigenfunctions are complete in a certain $L^2$ space.

The eigenvalues have been computed numerically in order to investigate the effect of various parameters on the stability of variable viscosity flows. The following are some of the key numerical results: (i) Variable viscosity flows can reduce the maximum growth rate by reducing the jumps in viscosity at the interfaces, but short waves become unstable; (ii) Short wave instability is dominated by the viscous gradient in the layer while long and intermediate wavelengths are dominated by the instability of the interfaces; (iii) The optimal viscous profile is the one which balances the interfacial instability with the instability of the layer; (iv) increasing interfacial tension decreases the viscous gradient of the optimal viscous profile; (v) A viscous profile which is exponential with respect to $\zeta$ is optimal among the types of profiles considered; (vi) $\sigma_{max}$ is a decreasing function of time. This is due to the instability of the interfaces decreasing with time while the layer instability remains relatively unchanged; and (vii) Variable viscosity flows allow for faster injection without making the flow more unstable.

\section*{Acknowledgment}
This work has been supported in part by the U.S. National Science Foundation grant DMS-1522782.

\bibliographystyle{plainnat}

\end{document}